\documentclass[letterpaper,11pt]{article}
\usepackage[utf8]{inputenc}
\usepackage[namelimits]{amsmath} 
\usepackage{forest}
\usepackage{fullpage}
\usepackage{color}
\usepackage{amsfonts}
\usepackage[left=1.0in,right=1.0in,top=1.0in,bottom=1.0in]{geometry}
\definecolor{Blue}{rgb}{0.1,0.1,0.8}
\usepackage{hyperref}
\hypersetup{
    linktocpage=true,
    colorlinks=true,				
    linkcolor=Blue,				
    citecolor=Blue,				
    urlcolor=Blue,			
}
\usepackage{amsthm}
\usepackage{amsmath,amsfonts}

\usepackage{amssymb}
\usepackage{soul}
\usepackage{graphicx}
\usepackage{enumitem}
\usepackage{multirow}
\usepackage{hhline}
\usepackage{caption}
\usepackage{subcaption}
\usepackage{colortbl}
\usepackage{thmtools,thm-restate}
\usepackage[ruled]{algorithm2e}
\usepackage{makecell}

\newtheorem{definition}{Definition}
\newtheorem{theorem}{Theorem}[section]
\newtheorem{lemma}[theorem]{Lemma}
\newcommand{\Max}{\mathrm{Max}}

\title{Almost Instance-optimal Clipping for Summation Problems in the Shuffle Model of Differential Privacy}
\date{}

\author{
Wei Dong\thanks{Carnegie Mellon University. {\tt wdong2@andrew.cmu.edu}. }
\and
Qiyao Luo\thanks{Hong Kong University of Science and Technology. {\tt qluoak@cse.ust.hk}. }
\and
Giulia Fanti\thanks{Carnegie Mellon University. {\tt gfanti@andrew.cmu.edu}. }
\and
Elaine Shi\thanks{Carnegie Mellon University. {\tt runting@gmail.com}. }
\and
Ke Yi\thanks{Hong Kong University of Science and Technology. {\tt yike@cse.ust.hk}. }
}

\begin{document}
\maketitle 

\begin{abstract}
Differentially private mechanisms achieving worst-case optimal error bounds (e.g., the classical Laplace mechanism) are well-studied in the literature. However, when typical data are far from the worst case, \emph{instance-specific} error bounds---which depend on the largest value in the dataset---are more meaningful. For example, consider the sum estimation problem, where each user has an integer $x_i$ from the domain $\{0,1,\dots,U\}$ and we wish to estimate $\sum_i x_i$. This has a worst-case optimal error of $O(U/\varepsilon)$, while recent work has shown that the clipping mechanism can achieve an instance-optimal error of $O(\max_i x_i \cdot \log\log U /\varepsilon)$. Under the shuffle model, known instance-optimal protocols are less communication-efficient. The clipping mechanism also works in the shuffle model, but requires two rounds: Round one finds the clipping threshold, and round two does the clipping and computes the noisy sum of the clipped data. In this paper, we show how these two seemingly sequential steps can be done simultaneously in one round using just $1+o(1)$ messages per user, while maintaining the instance-optimal error bound. We also extend our technique to the high-dimensional sum estimation problem and sparse vector aggregation (a.k.a. frequency estimation under user-level differential privacy).
\end{abstract}

\newpage

\section{Introduction}
\label{sec:intro}

The \textit{shuffle model} \cite{bittau2017prochlo,cheu2019distributed,erlingsson2019amplification} of \textit{differential privacy} (DP) is widely-studied in the context of DP computation over distributed data. 
The model has 3 steps: (1) Each client uses a randomizer $\mathcal{R}(\cdot)$ to privatize their data.
(2) A trusted shuffler randomly permutes the inputs from each client and passes them to an untrusted analyzer, which (3) conducts further analysis.  
Unlike the \textit{central model} of DP, where a trusted curator has access to all the data, the shuffle model provides stronger privacy protection by removing the dependency on a trusted curator.  Unlike the \textit{local model}, where clients send noisy results to the analyzer directly,
the addition of the trusted shuffler allows for a significantly improved privacy-accuracy tradeoff.  For problems like bit counting, shuffle-DP achieves an error of $O(1/\varepsilon)$ with constant probability\footnote{In Section \ref{sec:intro}, all stated error guarantees hold with constant probability.  We will make the dependency on the failure probability $\beta$ more explicit in later sections.}~\cite{ghazi2020private}, matching the best error of central-DP, while the smallest error achievable under local-DP is $O(\sqrt{n}/\varepsilon)$~\cite{bonawitz2017practical,damgaard2016unconditionally}.

The \textit{summation} problem, a fundamental problem with applications in statistics~\cite{kamath2019privately,biswas2020coinpress,huang2021instance}, data analytics~\cite{dong2022r2t,tao2020computing}, and machine learning such as the training of deep learning models~\cite{song2013stochastic,bassily2014private,abadi2016deep,agarwal2018cpsgd} and clustering algorithms~\cite{stemmer2021locally,stemmer2018differentially}, has been studied in many works under the shuffle model~\cite{balle2019privacy,cheu2019distributed,ghazi2019scalable,ghazi2020privateb,balle2020private,ghazi2021differentially}.
In this problem, each user $i \in [n]:=\{1,\dots,n\}$ holds an integer $x_i \in \{0,1,\dots,U\}$ and the goal is to estimate $\mathrm{Sum}(D):=\sum_i x_i$, where $D=(x_1,\dots,x_n)$. 
All of these works for sum estimation under shuffle-DP focus on achieving an error of $O(U/\varepsilon)$.  Such an error can be easily achieved under central-DP, where the curator releases the true sum after masking it with a Laplace noise of scale $U/\varepsilon$.  
In the shuffle-DP model, besides error, another criterion that should be considered is the communication cost.  The recent shuffle-DP summation protocol of \cite{ghazi2021differentially} both matches the error of central-DP and achieves optimal communication. More precisely, it achieves an error that is just $1+o(1)$ times that of the Laplace mechanism, while each user just sends in expectation $1+o(1)$ messages, each of a logarithmic number of bits.

However, in real applications (as well as in \cite{ghazi2021differentially}), $U$ must be set independently of the dataset; to account for all possible datasets, it should be conservatively large.  For instance, if we only know that the $x_i$'s are 32-bit integers, then $U=2^{32}-1$.   Then the error of $O(U/\varepsilon)$ could dwarf the true sum for most datasets.  Notice that sometimes some prior knowledge is available, and then a smaller $U$ could be used.  For example, if we know that the $x_i$'s are people's incomes, then we may set $U$ as that of the richest person in the world.  Such a $U$ is still too large for most datasets as such a rich person seldom appears in most datasets.

\textbf{Instance-Awareness.~}
The earlier error bound of $O(U/\varepsilon)$ can be shown to be optimal, but only in the worst case.
When typical input data are much smaller than $U$, an \emph{instance-specific} mechanism (and error bound) can be obtained---i.e., a mechanism whose error depends on the largest element of the dataset.
In the example of incomes above, an instance-aware mechanism would achieve an error proportional to the actual maximum income in the dataset.
This insight has recently been explored under the central model of DP \cite{andrew2019differentially, pichapati2019adaclip, mcmahan2017learning,huang2021instance,fang2022shifted,dong2022r2t}. 

A widely used technique for achieving instance-specific error bounds under central-DP is the \textit{clipping mechanism} \cite{andrew2019differentially, pichapati2019adaclip, mcmahan2017learning,huang2021instance}.  For some $\tau$, each $x_i$ is clipped to $\mathrm{Clip}(x_i,\tau) := \min(x_i,\tau)$.  Then we compute the sum of $\mathrm{Clip}(D,\tau) := \big(\mathrm{Clip}(x_i,\tau) \mid i =1,\dots,n\big)$ and add a Laplace noise of scale $O(\tau/\varepsilon)$.  Note that  the clipping introduces a (negative) bias of magnitude up to $\Max(D) \cdot |\{i \in [n] \mid x_i>\tau\}|$, where 
$\mathrm{Max}(D) := \max_i x_i$.
Thus, one should choose a good clipping threshold $\tau$ that balances the DP noise and bias.  Importantly, this must be done in a DP fashion, and this is where all past investigations on the clipping mechanism have been devoted. In the central model, the best error bound achievable is  $O\big(\mathrm{Max}(D)\cdot \log\log U/\varepsilon\big)$ \cite{dong2023universal}.\footnote{\cite{dong2023universal} achieves an error of $O(\mathrm{Max}(D)\log\log(\mathrm{Max}(D)))$ rather than the cited $O(\mathrm{Max}(D)\log\log U))$. The $\log\log(\mathrm{Max}(D))$ result is more meaningful for the unbounded domain setting where $U=\infty$, but in the shuffle-DP model, there is currently no known method can handle the unbounded domain case for any problem, including sum estimation. Our proposed mechanism also only supports the bounded domain case. Therefore, for simplicity, we ignored this minor difference and just cited the $\log\log U$ result.} For the real summation problem, such an error bound is considered (nearly) instance-optimal, since any DP mechanism has to incur an error of $\Omega(\mathrm{Max}(D))$ on either $D$ or $D-\{\mathrm{Max}(D)\}$ \cite{vadhan2017complexity}. 
The factor of $\log\log U/\varepsilon$ is known as the ``optimality ratio''.
It has been shown that the optimality ratio $O(\log\log U / \varepsilon)$ is the best possible in the case of $\delta=0$ ($\delta$ is a privacy parameter, see Section~\ref{sec:dp} for more details) \cite{dong2023universal}. Under the case of $\delta>0$, it is still an open question whether that ratio is optimal or not.
Notably, in the literature~\cite{asi2020instance,dong2022r2t,huang2021instance}, a polylog optimality ratio is often considered satisfactory enough and is named as ``instance-optimal'' and so far no known DP mechanism has a better optimality ratio.

As suggested in \cite{huang2021instance}, the clipping mechanism can be easily implemented in the shuffle model as well, but requiring two rounds.  The first round finds $\tau$.  Then we broadcast $\tau$ to all users.  In the second round, we invoke a summation protocol, e.g., the one in \cite{ghazi2021differentially}, on $\mathrm{Clip}(D,\tau)$. Two-round protocols are generally undesirable, not only because of the extra latency and coordination overhead, but also because they leak some information to the users ($\tau$ in this case, which is an approximation of $\Max(D)$). Note that the shuffle model, in its strict definition, only allows one-way messages from users to the analyzer (through the shuffler), so the users should learn nothing from each other.  Moreover, the central-DP mechanism for finding the optimal $\tau$ \cite{dong2023universal} does not work in the shuffle model.  Instead,  \cite{huang2021instance} uses the complicated range-counting protocol of \cite{ghazi2021power}.  This results in a sum estimation protocol that runs in two rounds, having an error of $\tilde{O}(\Max(D) \cdot \log^{3.5} U \sqrt{\log(1/\delta)})$\footnote{For any function $f$, $\tilde{O}(f):= f\cdot \mathrm{polylog}(f)$.}, and sends  $\tilde{O}\big(\log^{6}U\log(1/\delta)/\varepsilon\big)$ messages per user. 
Thus, this protocol is of only theoretical interest; indeed, no experimental results are provided in \cite{huang2021instance}.  

\textbf{Problem Statement.~}
Does there exist a practical, \textit{single-round} protocol for sum estimation under shuffle-DP that simultaneously: (1) matches the optimal central-DP error of $O(\Max(D) \cdot \log\log U / \varepsilon)$, and (2) requires $1+o(1)$ messages per user? 

\subsection{Our results}
\label{sec:results}


We answer this question in the affirmative, by presenting a new single-round shuffle-DP clipping protocol for the sum estimation problem.
At the core of our protocol is a technique that finds the optimal $\tau$ and computes the noisy clipped sum using $\tau$ at the same time.  This appears impossible, as the second step relies on the information obtained from the first. 
Our idea is to divide the data into a set of disjoint parts and do the estimations for each part independently.  This ensures we only pay the privacy and communication cost of one since each element will only be involved in one estimation. 
Based on these estimations, we can compute the noisy clipped sums for all the clipping thresholds $\tau=1,2,4,\dots, U$.  Meanwhile, we show that these noisy estimations already contain enough information to allow us to decide which $\tau$ is the best.
Besides solving sum aggregation, we show that using this protocol as a building block or deriving a variant of this idea can achieve state-of-the-art privacy-utility-communication tradeoffs for two other important summation problems.





\begin{table*}[ht]
\centering
\resizebox{1\columnwidth}{!}{
	\renewcommand\arraystretch{1.2}
\begin{tabular}{c|c|c|c|c}
\hline
\multicolumn{3}{c|}{Mechanism} & Error & Average messages sent by each user 
\\
\hline
\hline
\multirow{9}{*}{\makecell[c]{$1$-D \\ Sum}} & \multirow{5}{*}{\makecell[c]{Prior\\ work}} & \cite{ghazi2021differentially} & $O\big(U/\varepsilon\big)$ & $1+o(1)$
\\
\cline{3-5} 
& & \cite{balle2020private} &$O\big(U/\varepsilon\big)$ & $O(1)$
\\
\cline{3-5} 
& & \multirow{3}{*}{\cite{huang2021instance} + \cite{ghazi2021differentially} + \cite{ghazi2021power}} & \multirow{3}{*}{\makecell[c]{$\tilde{O}\big(\mathrm{Max}(D)\cdot\log^{3.5} U$\\$\sqrt{\log(1/\delta)}/\varepsilon\big)$}} & \multirow{3}{*}{\makecell[c]{\textbf{Round} $1$: $\tilde{O}\big(\log^6 U\cdot \log\big(1/\delta\big)/\varepsilon\big)$ \\ \textbf{Round} $2$: $1+o(1)$}}
\\
& & &  & 
\\
& & &  & 
\\
\cline{2-5}
&\multirow{2}{*}{\makecell[c]{Our\\ result}} & Theorem & \multirow{2}{*}{\makecell[c]{$O\big(\mathrm{Max}(D)\cdot\log\log U/\varepsilon\big)$}} & \multirow{2}{*}{\makecell[c]{$1+o(1)$}}
\\
& & \ref{th:sum} & &
\\
\hhline{|~|----|}
& \multicolumn{2}{c|}{Best result under} & \multirow{2}{*}{$O\big(\mathrm{Max}(D)\cdot \log\log U/\varepsilon\big)$} & \cellcolor[gray]{0.9}
\\
& \multicolumn{2}{c|}{central model~\cite{dong2023universal}} & &\cellcolor[gray]{0.9}
\\
\hline
\multirow{11}{*}{\makecell[c]{$d$-D \\ Sum}} & \multirow{5}{*}{\makecell[c]{Prior\\ work}} & \multirow{2}{*}{\cite{huang2021instance}} & \multirow{2}{*}{\makecell[c]{$\tilde{O}\big(U_{\ell_2}\sqrt{d\log(n)\log(1/\delta)}/\varepsilon\big)$}} & \multirow{2}{*}{\makecell[c]{$d+\tilde{O}(d^{1.5}\log^{1.5}(1/\delta)/(\varepsilon\sqrt{n})\big)$}}
\\
& & & &
\\
\cline{3-5}
& & \multirow{3}{*}{\cite{huang2021instance} + \cite{ghazi2021differentially} +\cite{ghazi2021power}} & \multirow{3}{*}{\makecell[c]{$\tilde{O}\big(\mathrm{Max}_{\ell_2}(D)\cdot \big(\sqrt{d\log(nd)\log(1/\delta)}$
\\
$+\log^{3.5}U_{\ell_2}\cdot\sqrt{\log(1/\delta)}\big)/\varepsilon\big)$}} & \multirow{3}{*}{\makecell[c]{\textbf{Round} $1$: $\tilde{O}\big(\log^6 U_{\ell_2}\cdot \log (1/\delta)/\varepsilon\big)$ \\ \textbf{Round} $2$: $d+\tilde{O}(d^{1.5}\log^{1.5}(1/\delta)/(\varepsilon\sqrt{n})\big)$}}
\\
& & & &
\\
& & & &
\\
\cline{2-5}
& \multirow{3}{*}{\makecell[c]{Our\\ result}} & \multirow{3}{*}{\makecell[c]{Theorem\\\ref{th:high_dim_sum}}} & \multirow{3}{*}{\makecell[c]{$O\big(\mathrm{Max}_{\ell_2}(D)\cdot\log(d\log U_{\ell_2})$\\$ \cdot\sqrt{d\log(nd)\log(1/\delta)}/\varepsilon\big)$}} & \multirow{3}{*}{\makecell[c]{$d+\tilde{O}(d^{1.5}\log^{1.5}(1/\delta)/(\varepsilon\sqrt{n})\big)$}}
\\
& & & &
\\
& & & &
\\
\hhline{|~|----|}
&
\multicolumn{2}{c|}{\multirow{3}{*}{\makecell[c]{Best result under\\ central model~\cite{dong2023better}}}}& \multirow{3}{*}{\makecell[c]{$O\big(\mathrm{Max}_{\ell_2}(D)\cdot \big(\sqrt{d\log(1/\delta)}$
\\
$+\log\log U_{\ell_2}\big)/\varepsilon\big)$}} & \cellcolor[gray]{0.9}
\\
& \multicolumn{2}{c|}{} & &\cellcolor[gray]{0.9}
\\
& \multicolumn{2}{c|}{} & &\cellcolor[gray]{0.9}
\\
\hline
\multirow{5}{*}{\makecell[c]{Sparse \\ Vector\\Aggregation}} & \multirow{2}{*}{\makecell[c]{Our\\ result}} & Theorem & \multirow{2}{*}{\makecell[c]{$\tilde{O}(\Max_{\ell_2}(D)\cdot \log d \sqrt{\log(1/\delta)}/ \varepsilon)$}} & \multirow{2}{*}{\makecell[c]{$\|x_i\|_1 + 1 +O(d^{1.5}\log d \log^{1.5}(1/\delta)/(\varepsilon n))$}}
\\
& & \ref{th:sparse_vector} & 
\\
\hhline{|~|----|}
&
\multicolumn{2}{c|}{\multirow{3}{*}{\makecell[c]{Best result under\\ central model~\cite{dong2023better}}}}& \multirow{3}{*}{\makecell[c]{$O\big(\mathrm{Max}_{\ell_2}(D)\cdot \big(\sqrt{\log(1/\delta)}$
\\
$+\log\log U_{\ell_2}\big)/\varepsilon\big)$}} & \cellcolor[gray]{0.9}
\\
& \multicolumn{2}{c|}{} & &\cellcolor[gray]{0.9}
\\
& \multicolumn{2}{c|}{} & &\cellcolor[gray]{0.9}
\\
\hline
\end{tabular}
}
\caption{Comparison between our results and prior works on sum estimation, high-dimensional sum estimation, and sparse vector aggregation under shuffle model, where we use the absolute error, $\ell_2$ error, and $\ell_{\infty}$ error as the error metrics.
Each message contains $O(\log U+\log d+\log n)$ bits.  The mechanism without an indicator for the round runs in a single round.
Our communication cost for 1-D Sum requires the condition $n=\omega(\log^2 U)$.
}
\label{tab:comparison}
\end{table*}

\subsubsection{Contributions}
Our contributions are threefold, summarized in Table \ref{tab:comparison} and below:

\textbf{(1) Sum estimation.}
For the vanilla sum estimation problem\footnote{Note that although we focus on the integer domain $\{0,\dots,U\}$, our protocol easily extends to the real summation problem, where each value $x_i$ is a real number from $[0,1]$, by discretizing $[0,1]$ into $U$ buckets of width $1/U$.  This incurs an extra additive error of $O(n/U)$.  Thanks to the double logarithmic dependency on $U$, we could set $U$ sufficiently large (e.g., $U=n^{\log n}$) to make this additive error negligible while keeping the $O(\Max(D)\cdot \log\log n / \varepsilon)$ error bound. 
}, we present a single-round protocol (Section \ref{sec:strawman} and \ref{sec:sum}) that achieves the optimal error of $O(\Max(D) \cdot \log\log U / \varepsilon)$, which improves the error rate of $O(U/\varepsilon)$ from \cite{ghazi2021differentially} exponentially in $U$.
More importantly, we have $1+o(1)$ messages per client when $n=\omega(\log^2 U)$ (see Theorem~\ref{th:sum} for more details), a criterion typically met in most common regimes.

\smallskip

\textbf{(2) High-dimensional sum estimation.}
Next, we consider the sum estimation problem in high dimensions, which has been extensively studied in the machine learning literature 
under central DP \cite{huang2021instance,biswas2020coinpress,kamath2019privately}.  Here, each $x_i$ is a vector with integer coordinates taken from the $d$-dimensional ball of radius $U_{\ell_2}$ centered at the origin, and we wish to estimate $\mathrm{Sum}(D)$ with small $\ell_2$ error. 

The literature for this problem exhibits similar patterns to the 1D summation problem. 
Under central-DP, the state-of-the-art mechanism achieves an error proportional to $\sqrt{d}\cdot\mathrm{Max}_{\ell_2}(D)$, where $\Max_{\ell_2}(D):=\max_i \|x_i\|_2$ \cite{dong2023better}.  Generalizing the argument in the 1D case, $\Max_{\ell_2}(D)$ is an instance-specific lower bound for $d$-dimensional sum estimation, and the factor $\sqrt{d}$ is also optimal~\cite{huang2021instance}.
For shuffle-DP, \cite{huang2021instance} presented a one-round protocol achieving an error proportional to $\sqrt{d}\cdot U_{\ell_2}$ (i.e., not instance-specific) with $d+\tilde{O}(d^{1.5}\log^{1.5}(1/\delta)\\/(\varepsilon\sqrt{n})\big)$ message complexity. 
\cite{huang2021instance} observed that 
a two-round clipping mechanism can be used 
to achieve an instance-specific error, but as in the 1D case, this incurs high polylogarithmic factors in both the optimality ratio and the message complexity. 

In Section \ref{sec:highD}, we propose our single-round protocol for high-dimensional summation by treating our 1D summation protocol as a black box: we first do a rotation over the space, and invoke our 1D protocol in each dimension.  This approach has the same instance-optimal error as the central-DP up to polylogarithmic factors, and achieves the same message complexity as the existing worst-case error protocol.

\smallskip \textbf{(3) Sparse vector aggregation.~}
As the third application of our technique, we study the sparse vector aggregation problem.  This problem is the same as the high-dimensional sum estimation problem, except that (1) each $x_i$ is now a binary vector in $\{0,1\}^d$, (2) the $x_i$'s are sparse, i.e., $\|x_i\|_1 = \|x_i\|_2^2 \ll d$, and (3) we aim at an $\ell_\infty$ error.  This problem is also known as the \textit{frequency estimation} problem under \textit{user-level DP}, where each user contributes a set of elements from $[d]$, and we wish to estimate the frequency of each element.  For this problem, people are more interested in the $\ell_\infty$ error since we would like each frequency estimate to be accurate.

Under central-DP, the state-of-the-art algorithm  already achieves $\ell_\infty$ error with $\mathrm{Max}_{\ell_2}(D)$ \cite{dong2023better}. Under shuffle-DP, there is no known prior work on this problem.   Our high-dimensional sum estimation protocol can solve the problem, but it does not yield a good $\ell_\infty$ error and incurs a message complexity of at least $d$ per user, even if the user has only a few elements.

In Section \ref{sec:sparse}, we present our one-round sparse vector aggregation protocol.  This protocol can be regarded as a variant of our 1D sum protocol, where we divide the data per its sparsity. This has error of $\Max_{\ell_2}(D)$, and sends $\|x_i\|_1 + 1 +O(d^{1.5}\log d \log^{1.5}(1/\delta)/(\varepsilon n))$ messages for user $i$. Note that this $\ell_\infty$ error implies an $\ell_2$ error that is $\sqrt{d}$ times larger (but not vice versa), so it also matches the high-dimensional sum estimation protocol in terms of $\ell_2$ error.  Furthermore, it exploits the sparsity of each $x_i$ in the message complexity.  It remains an interesting open problem if the extra $O(d^{1.5}\log d \log^{1.5}(1/\delta)/(\varepsilon n))$ term can be reduced.

\section{Related Work}

For sum estimation under central-DP, the worst-case optimal error $O(U/\varepsilon)$ can be easily achieved by the \textit{Laplace mechanism}. Many papers have studied how to obtain instance-specific error, i.e., an error depending on $\mathrm{Max}(D)$~\cite{asi2020instance,andrew2019differentially, pichapati2019adaclip, mcmahan2017learning,huang2021instance,dong2022r2t,fang2022shifted,dick2023subset,dong2023universal}.  Most of these works rely on the clipping mechanism~\cite{asi2020instance,andrew2019differentially, pichapati2019adaclip, mcmahan2017learning,huang2021instance,dong2022r2t,dick2023subset,dong2023universal}. 
Similarly, for the high-dimensional sum aggregation problem, existing approaches have achieved instance-specific error by using the clipping mechanism~\cite{kamath2019privately,kamath2020private,biswas2020coinpress,huang2021instance,dong2023better};  such mechanisms also yield an $\ell_{\infty}$ error of $\mathrm{Max}_{\ell_2}(D)$ for sparse vector aggregation.

In the shuffle-DP setting, for sum estimation, two settings are used. In the \textit{single-message} setting, each user sends one message. Here, \cite{balle2019privacy} achieve an error of $O(Un^{1/6})$ and further show that this is worst-case optimal.  In the \textit{multi-message} setting, where each user is allowed to send multiple messages, most prior works try to achieve the worst-case optimal error while minimizing communication costs.  Cheu~et al.~\cite{cheu2019distributed} first achieved an error of $O(U\sqrt{\log(1/\delta)}/\varepsilon)$ with $O(\sqrt{n})$ messages sent per user.  Then, \cite{ghazi2019scalable} achieved the same error but reduced the number of messages per user to $O(\log(n))$. \cite{ghazi2020privateb} and \cite{balle2020private} further improved the error to $O(U/\varepsilon)$ with constant messages. 
Recently, \cite{ghazi2021differentially} reduced that communication to $1+o(1)$ messages per user. 
We aim to obtain instance-optimal error, while keeping the $1+o(1)$ per-client message complexity of \cite{ghazi2021differentially}.

\section{Preliminaries}
We use the following notation: $\mathbb Z$ is the domain of all integers, $\mathbb Z_{\geq 0}$ non-negative integers, and $\mathbb Z_+$ positive integers.
Let $D = (x_1,x_2,\dots,x_n)$, where user $i$ holds an integer $x_i$ from $\{0\}\cup [U]$.  
For simplicity, we assume that $U$ is a power of $2$.
We would like to estimate $\mathrm{Sum}(D) = \sum_i x_i$.  For brevity, we often interpret $D$ as a multiset, and $D\cap [a,b]$ denotes the multiset of elements of $D$ that fall into $[a,b]$.  We introduce two auxiliary functions: $\mathrm{Count}(D)$ is the cardinality of $D$ (duplicates are counted); $\mathrm{Max}(D,k)$ is the $k$th largest value of $D$, or more precisely,
\[\mathrm{Max}(D,k) := \max \Big\{t:\mathrm{Count}(D\cap [t,U])\geq k\Big\}.\]

\subsection{Differential Privacy}
\label{sec:dp}

\begin{definition}[Differential privacy] 
\label{def:diff_privacy}
For $\varepsilon, \delta > 0$, an algorithm $\mathcal{M}:\mathcal{X}^n\rightarrow \mathcal{Y}$ is $(\varepsilon, \delta)$-differentially private (DP) if for any neighboring instances $D\sim D'$ (i.e., $D$ and $D'$ differ by a single element), $\mathcal{M}(D)$ and $\mathcal{M}(D')$ are $(\varepsilon,\delta)$-indistinguishable, i.e., for any subset of outputs $Y\subseteq\mathcal{Y}$,
\[\mathsf{Pr}[\mathcal{M}(D)\in Y] \le e^\varepsilon \cdot \mathsf{Pr}[\mathcal{M}(D') \in Y] + \delta.\]
\end{definition}
The privacy parameter $\varepsilon$ is typically between 0.1 and 10, while $\delta$ should be much  smaller than $1/n$.

All DP models can be captured by the definition above by appropriately defining $\mathcal{M}(D)$. In \textit{central}-DP,  $\mathcal{M}(D)$ is just the output of data curator;
in \textit{local}-DP, the local randomizer $\mathcal{R}:\mathcal{X}\rightarrow \mathcal{Z}$ outputs a message in $\mathcal{Z}$, and $\mathcal{M}(D)$ is defined as the vector $\big(\mathcal{R}(x_1),\mathcal{R}(x_2),\dots, \mathcal{R}(x_n)\big)$; in \textit{shuffle}-DP,  $\mathcal{R}:\mathcal{X} \rightarrow \mathbb{N}^{\mathcal{Z}}$ outputs a multiset of messages and $\mathcal{M}(D)$ is the (multiset) union of the $\mathcal{R}(x_i)$'s. 

DP enjoys the following properties regardless of the specific model:

\begin{lemma} [Post Processing \cite{dwork2014algorithmic}]
\label{lm:post_processing_dp}
If $\mathcal{M}$ satisfies $(\varepsilon,\delta)$-DP and $\mathcal{M}'$ is any randomized mechanism, then $\mathcal{M}'(\mathcal{M}(D))$ satisfies $(\varepsilon,\delta)$-DP.
\end{lemma}

\begin{lemma}
[Sequential Composition \cite{dwork2014algorithmic}]
\label{lm:sequential_composition_dp}
If $\mathcal{M}$ is a (possibly adaptive) composition of differentially private mechanisms $\mathcal{M}_1,  \ldots, \mathcal{M}_k$, where each $\mathcal{M}_i$ satisfies $(\varepsilon, \delta)$-DP, then $\mathcal{M}$ satisfies $(\varepsilon', \delta')$-DP, where
\begin{enumerate}
\item $\varepsilon' = k\varepsilon$ and $\delta'=k\delta$; [Basic Composition]
\item $\varepsilon' = \varepsilon\sqrt{2k\log\frac{1}{\delta''}} + k\varepsilon(e^{\varepsilon} - 1)$ and $\delta'=k\delta+\delta''$ for any $\delta''>0$. [Advanced Composition]
\end{enumerate}
\end{lemma}

\begin{lemma}
[Parallel Composition \cite{mcsherry2009privacy}]
\label{lm:parallel_composition_dp}
Let $\mathcal{X}_1,\dots, \mathcal{X}_k$ each be a subdomain of $\mathcal{X}$ that are pairwise disjoint, and let each $\mathcal{M}_i:\mathcal{X}_i^n \rightarrow \mathcal{Y}$ be an $(\varepsilon,\delta)$-DP mechanism.  Then $\mathcal{M}(D):= (\mathcal{M}_1(D\cap \mathcal{X}_1), \dots, \mathcal{M}_k(D\cap \mathcal{X}_k))$ also satisfies $(\varepsilon, \delta)$-DP.
\end{lemma}

\subsection{Sum Estimation in Central-DP}
\label{sec:sum_central}

In central-DP, one of the most widely used DP mechanisms is the Laplace mechanism:

\begin{lemma}[Laplace Mechanism]
\label{lm:lap}
Given any query $Q:\{0,1,\dots,U\}^n\rightarrow \mathbb{R}$, the \textit{global sensitivity} is defined as $\mathrm{GS}_Q = \max_{D\sim D'} \big|Q(D)-Q(D')\big|$.  
The mechanism $\mathcal{M}(D)=Q(D)+\mathrm{GS}_Q/\varepsilon \cdot \mathrm{Lap}(1)$
preserves $(\varepsilon,0)$-DP, where $\mathrm{Lap}(1)$ denotes a random variable drawn from the unit Laplace distribution.
\end{lemma}

For the sum estimation problem, $\mathrm{GS}_Q = U$, which means that the Laplace mechanism yields an error of $O(U/\varepsilon)$.  As mentioned in Section \ref{sec:intro}, although such an error bound is already worst-case optimal, it is not very meaningful when typical data are much smaller than $\mathrm{GS}_Q = U$.  

\textbf{Clipping mechanism} The clipping mechanism has been widely used to achieve an instance-specific error bound depending on $\mathrm{Max}(D)$. It first finds a clipping threshold $\tau$, and then applies the Laplace mechanism on $\mathrm{Clip}(D,\tau)$ with $\mathrm{GS}_Q = \tau$.  The clipping introduces a  bias of $\Max(D) \cdot |\{i\in [n] \mid x_i>\tau\}|$, so it is important to choose a $\tau$ that balances the DP noise and bias.  In the central model, the best result is \cite{dong2023universal}, which
finds a $\tau$ between $\Max(D,\log\log U/\varepsilon)$ and  $2\cdot \Max(D)$.  Plugging this $\tau$ into the clipping mechanism yields an error of $O\big(\mathrm{Max}(D)\cdot \log\log U/\varepsilon\big)$.

\begin{algorithm}[t]
\LinesNumbered 
\SetNoFillComment
\KwIn{$x_i$, $\varepsilon$, $\delta$, $n$, $U$, $\lambda$, $\zeta$}
$U',x_i'\gets U,x,$\;
\uIf{$U>\sqrt{n/\zeta}$}{
    \tcc{Randomized rounding of $x_i$}
    $B\gets \big\lceil U/(\sqrt{n/\zeta})\big\rceil$\;
    $U'\gets \sqrt{n}$\; 
    $p\gets \big\lceil x_i/B\big\rceil - x_i/B$\;
    $x_i'\gets \begin{cases}
        \big\lceil x_i/B\big\rceil & \text{with probability $p$}
        \\
        \big\lceil x_i/B\big\rceil-1 & \text{with probability $1-p$}
    \end{cases}$\;
}
$S_i \gets \{\}$\;
\If{$x_i'\neq 0$}{
    Add $x_i'$ into $S_i$\;
}
\tcc{Sample a vector from $\mathcal{P}$}
$(z^{-U'},\dots,z^{-1},z^{1},\dots,z^{U'})\sim\mathcal{P}(\varepsilon,\delta,n,\lambda,U')$\;
\For{$j\gets -U',-U'+1,\dots,-1,1,\dots,U'-1,U'$}{      Add $z_j$ copies of $j$ into $S_i$\;
}
Send $S_i$\;
\caption{Randomizer of $\mathrm{BaseSumDP}$~\cite{ghazi2021differentially}.}
\label{alg:baseline_randomizer}
\end{algorithm}

\begin{algorithm}[t]
\LinesNumbered 
\KwIn{$R=\cup_i S_i$ with $S_i$ from user $i$, $\varepsilon$, $\delta$, $n$, $U$, $\lambda$, $\zeta$}
$\widetilde{\mathrm{Sum}}(D)\gets\sum_{y\in R}y$\;
\uIf{$U>\sqrt{n}$}{
    $B\gets \big\lceil U/(\sqrt{n/\zeta})\big\rceil$\;
    $\widetilde{\mathrm{Sum}}(D)\gets \widetilde{\mathrm{Sum}}(D)\cdot B$\;
}
\Return $\widetilde{\mathrm{Sum}}(D)$;
\caption{Analyzer of $\mathrm{BaseSumDP}$~\cite{ghazi2021differentially}.}
\label{alg:baseline_analyzer}
\end{algorithm}

\subsection{Sum Estimation in Shuffle-DP}

In shuffle-DP, the state-of-the-art protocol for sum estimation is proposed by Ghazi et al.~\cite{ghazi2021differentially} and achieves an error that is $1+o(1)$ times that of the Laplace mechanism and each user sends $1+o(1)$ messages in expectation.  Both are optimal (in the worst-case sense) up to lower-order terms.  We briefly describe their protocol below as it will also be used in our protocols.  

Each user $i$ first sends $x_i$ if it is non-zero. To ensure privacy, users additionally send noises drawn from $\{-U,\dots, U\} - \{0\}$ based on an ingeniously designed distribution $\mathcal{P}$, such that most noises cancel out while the remaining noises add up to a random variable drawn from the discrete Laplace distribution\footnote{The Discrete Laplace distribution with scale $s$ has a probability mass function $\frac{1-e^{-1/s}}{1+e^{-1/s}}\cdot e^{-|k|/s}$ for each $k\in \mathbb{Z}$.} with scale $(1-\lambda)\varepsilon$ for some parameter $\lambda$. 
The cancelled out noises are meant to flood the messages containing the true data $x_i$ so as to ensure $(\lambda \varepsilon, \delta)$-DP.  Thus, the entire protocol satisfies $(\varepsilon,\delta)$-DP.  

For a large $U$, their protocol should be applied after reducing the domain size to $\sqrt{n/\zeta}$ for some $\zeta=o(1)$.  More precisely, we first randomly round each $x_i$ to a multiple of $\frac{U}{\sqrt{n/\zeta}}$. This introduces an additional error of $O\big(\sqrt{\zeta} U\big)$, which is a lower-order term in the error bound.  Meanwhile, it reduces the noise messages to $O\big(\log^2 n\log(1/\delta)/(\varepsilon\lambda\sqrt{\zeta n})\big)=o(1)$.
The detailed randomizer and analyzer are given in Algorithm~\ref{alg:baseline_randomizer} and \ref{alg:baseline_analyzer}. 
The analyzer obviously runs in $O(n)$ time, and they show how to implement the randomizer in time  $O\big(\min(U,\sqrt{n})\big)$.  The following lemma summarizes their protocol:

\begin{lemma}
\label{lm:baseline}
Given any $\varepsilon>0$, $\delta>0$, $n$, $U$, any $\lambda$ and any $\zeta$, $\mathrm{BaseSumDP}$ solves the sum estimation problem under shuffle-DP with the following guarantees:
\begin{enumerate}
    \item The messages received by the analyzer satisfy $(\varepsilon, \delta)$-DP;
    \item With probability at least $1-\beta$, the error is bounded by $(\zeta+\frac{1}{\varepsilon(1-\lambda)})\cdot U \ln(2/\beta)$;
    \item In expectation, each user sends 
    \[\mathbf{I}(x_i\neq 0)+ O\Big({\log^2 n \log(1/\delta)\over\varepsilon \lambda\sqrt{\zeta n}}\Big)\]
    messages each of $O\Big(\min\big(\log n,\log U\big)\Big)$ bits. 
\end{enumerate}
\end{lemma}

\noindent\textbf{Remark:} Setting $\lambda,\zeta=o(1)$ yields an error of $1+o(1)$ error and $1+o(1)$ messages.
In this paper, we will invoke $\mathrm{BaseSumDP}$ with $\lambda = 0.1$ and $\zeta = \min(0.1,\frac{0.1}{\varepsilon})$.  With this setting, the error  bound is $1.3\cdot U\ln(2/\beta)/\varepsilon$ and the message complexity is still $1+o(1)$. 




\medskip
Combining the clipping technique with $\mathrm{BaseSumDP}$ immediately leads to a two-round protocol in shuffle-DP: In round one, we find $\tau$; in round two, we invoke $\mathrm{BaseSumDP}$ on $\mathrm{Clip}(D,\tau)$.  This approach was suggested in \cite{huang2021instance}.  However, since the optimal central-model $\tau$-finding algorithm \cite{dong2023universal} cannot be used in the shuffle model, \cite{huang2021instance} instead used the complicated range-counting protocol of \cite{ghazi2021power} to find a $\tau$ such that
\begin{equation}
\label{eq:rank_error}
\mathrm{Max}(D,k)\leq \tau \leq {\mathrm{Max}}(D)
\end{equation}
for some $k = \tilde{O}\big(\log^{3.5}U\sqrt{\log(1/\delta)}/\varepsilon\big)$.  Plugging this $\tau$ into the clipping mechanism yields an error of $O(\Max(D) \cdot k)$.  The message complexity of their protocol is $\tilde{O}\big(\log^{6}U\log(1/\delta)/\varepsilon\big)$, dominated by the range-counting protocol \cite{ghazi2021power}.

\section{A Straw-man One-Round Protocol}
\label{sec:strawman}
We first present a simple one-round shuffle-DP protocol for the sum estimation problem.  Although it does not achieve either the desired error or communication rates from Section \ref{sec:intro}, it provides a foundation for our final solution.

\subsection{Domain Compression}
\label{sec:compression}

Our first observation is it is not necessary to consider all possible $\tau\in \{0\}\cup [U]$.  Instead, we only need to consider $\tau\in\{0,1,2,4,\dots,U\}$.
Specifically, we map the dataset $D$ to\footnote{All $\log$ have base 2.  Specially, define $\log(0):=-1$ and $2^{-1}:=0$.} 
\[\overline{D} = \Big\{\big\lceil\log(x_1)\big\rceil,\big\lceil\log(x_2)\big\rceil,\dots,\big\lceil\log(x_n)\big\rceil\Big\}.\]
Note that this compresses the domain from $\{0\}\cup[U]$ to $\{-1, 0\} \cup [\log U]$. 
After compressing the domain size from $U+1$ to $\log U + 2$, running the round-one protocol of \cite{huang2021instance} on $\overline{D}$ can now find a $\bar{\tau}$ such that
\begin{equation}
\label{eq:rank_error_compress}
\mathrm{Max}(\overline{D},k)\leq \bar{\tau} \leq {\mathrm{Max}}(\overline{D}),
\end{equation}
for some $k = \tilde{O}\big((\log\log U)^{3.5}\sqrt{\log(1/\delta)}/\varepsilon\big)$.  

In the second round, we use $\tau = 2^{\bar{\tau}}$ as the clipping threshold and invoke $\mathrm{BaseSumDP}$ on $D$.  Note that we always have $\tau \le 2\cdot \Max(D)$.  Furthermore, $D$ contains at most $k$ elements that are strictly larger than $\tau$, so the clipping mechanism yields an error of $O(\Max(D)\cdot k)$.

In addition to reducing the error, domain compression also reduces the message complexity of the first round from $\tilde O((\log U)^6 \log(1/\delta)/\epsilon)$ to $\tilde{O}((\log\log U)^6 \log(1/\delta)/\varepsilon)$.  The message complexity of the second round is the same as that of $\mathrm{BaseSumDP}$, i.e., $1+o(1)$.

\subsection{Try All Possible $\tau$}
\label{sec:one-round}

The domain compression technique narrows down the possible values for $\tau$ to just $\log U + 2$.  This allows us to try all possible $\tau$ simultaneously.  We can run $\log U + 2$ instances of $\mathrm{BaseSumDP}$, each with a different $\tau=0,1,2,4,\dots, U$. 
That is, each client $x_i$ runs $\mathrm{BaseSumDP}$ $\log U +2$ times, each time clipping its data $x_i$ with a different threshold before the randomizer protocol, and the analyzer computes $\log U +2$ different sums, one for each threshold.
All these are done concurrently to the protocol for finding $\bar{\tau}$.  Finally, the analyzer will return the output of the  $\mathrm{BaseSumDP}$ instance that has been executed with the correct $\tau = 2^{\bar{\tau}}$.  
However, the $\log U + 2 $ instances of $\mathrm{BaseSumDP}$ must split the privacy budget using sequential composition\footnote{``Sequential composition'' refers to privacy; all these instances are still executed in parallel in one round. }.  More precisely, we run each instance with privacy budget $\varepsilon/(2(\log U +2))$, while reserving the other $\varepsilon/2$ privacy budget for finding $\bar{\tau}$.  Thus, the $\mathrm{BaseSumDP}$ instance with clipping threshold $\tau$ must inject a DP noise of scale $O(\tau \log U / \varepsilon)$.  The clipping still introduces a bias of $O(\Max(D)\cdot k)$, so the total error becomes $\tilde{O}(\Max(D)\cdot (\log U + \sqrt{\log(1/\delta)}) /\varepsilon)$.

In terms of the message complexity, these $O(\log U)$ $\mathrm{BaseSumDP}$ instances together send $O(\log U)$ messages per user, in addition of the $O((\log\log U)^6 \log(1/\delta)/\varepsilon)$ messages for finding $\bar{\tau}$.  So the message complexity is now $O(\log U + \log(1/\delta)/\varepsilon)$.

Simple tweaks to this straw-man solution do not give the desired properties.  For instance, one may compress the domain to $\{0,1,c,c^2,\dots\}$ for some $c\ge 2$.  This lowers the message complexity to $O(\log_c U + \log(1/\delta)/\varepsilon)$ and each $\mathrm{BaseSumDP}$ instance has a privacy budget of $\varepsilon/\log_c U$.  But now $\tau$ may be as large as $c\cdot \Max(D)$, so the error increases to $O(\Max(D) \cdot c \log_c U/\varepsilon)$.  Thus, new ideas are needed to achieve our desiderata in a one-round protocol.





\section{Our Protocol}
\label{sec:sum}
In this section, we present our single-round protocol that achieves both optimal error and message complexity. 


\subsection{Domain Partitioning}

We see that the $O(\log U)$ factor blowup in the error of the straw-man solution is due to the $\log U+2$ $\mathrm{BaseSumDP}$ instances 
splitting the privacy budget using sequential composition.  In order to avoid the splitting, our idea is to partition the domain and then use parallel composition.  More precisely, we partition the domain into $\log U + 1$ disjoint sub-domains: $[1,1]$, $[2,2]$, $[3,4]$, $[5,8]$, $\dots$, $[U/2+1,U]$. It is clear that, for any $D$, we have 
\[\mathrm{Sum}(D) = \sum_{j=0}^{\log U}\mathrm{Sum}\big(D\cap[2^{j-1}+1,2^j]\big).\]
Furthermore, for any $\tau = 1,2,4,\dots,U$, the clipped sum is precisely the sum in the first $\log \tau +1$ sub-domains:
\[
\mathrm{Sum}\Big(\mathrm{Clip}(D,\tau)\Big) = 
\sum_{j=0}^{\log \tau}\mathrm{Sum}\big(D\cap[2^{j-1}+1,2^j]\big).
\]

\begin{figure*}[t]
\centering
\includegraphics[width=1\textwidth]{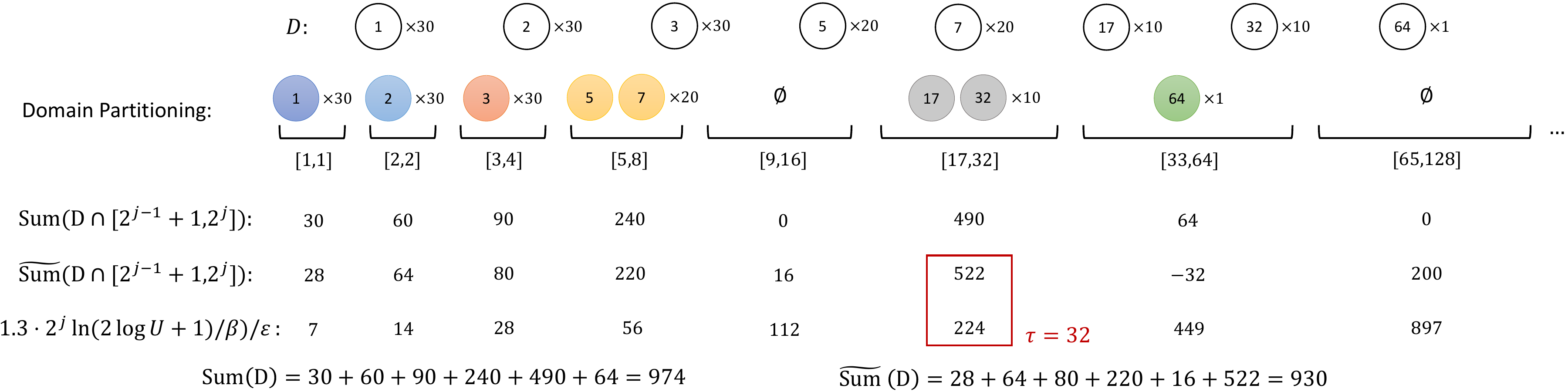}
    \caption{An illustration of our protocol for sum estimation. $U=2^{10}$, $\varepsilon=1$, and $\beta = 0.1$.}
    \label{fig:illustration}
\end{figure*}

Therefore, it suffices to estimate $\mathrm{Sum}\big(D\cap[2^{j-1}+1,2^j]\big)$ for each $j\in {0,1,\dots,\log U}$.  Importantly, since these sub-domains are disjoint, parallel composition can be applied and we can afford a privacy budget of $\varepsilon$ on each sub-domain.  We thus run a $\mathrm{BaseSumDP}$ instance on each $D\cap[2^{j-1}+1,2^j]$, which returns an estimate 
\begin{align*} &\widetilde{\mathrm{Sum}}\big(D\cap[2^{j-1}+1,2^j]\big) := \\
&\mathrm{Sum}\big(D\cap[2^{j-1}+1,2^j]\big)  + \mathrm{Lap}(2^j/\varepsilon).
\end{align*}
Then for any $\tau=1,2,4,\dots,U$, we estimate $\mathrm{Sum}(\mathrm{Clip}(D,\tau))$ as 
\[\widetilde{\mathrm{Sum}}\big(D\cap[1,\tau]\big) = \sum_{j=0}^{\log\tau}\widetilde{\mathrm{Sum}}\big(D\cap[2^{j-1}+1,2^j]\big).\]
Importantly, the total noise level in $\widetilde{\mathrm{Sum}}\big(D\cap[1,\tau]\big)$ is still bounded by $O(\tau/\varepsilon)$, as the noise levels from the sub-domains form a geometric series.  This ensures that the DP noise is bounded by $O(\Max(D)/\varepsilon)$, as long as we choose a $\tau\le 2\cdot \Max(D)$.  

Meanwhile, this domain partitioning lowers the total message complexity of all the $\log U +1$ $\mathrm{BaseSumDP}$ instances to $1+o(1)$.   This is because after domain partitioning, each user has a nonzero input only in one sub-domain, and the $\mathrm{BaseSumDP}$ protocol sends out $o(1)$ message when $x_i=0$.

\begin{algorithm*}[t]
\LinesNumbered 
\SetNoFillComment
\KwIn{$x_i$, $\varepsilon$, $\delta$, $\beta$, $n$, $U$}
\For{$j\gets 0,1,2,\dots,\log U$}{
    \tcc{The messages for estimating $\mathrm{Sum}\big(D\cap [2^{j-1}+1,2^j]\big)$}
    $S_{i}^{[2^{j-1}+1,2^j]}\gets \mathrm{Randomizer}\Big(x_i\cdot \mathbf{I}\big(x_i\in [2^{j-1}+1,2^j]\big),\varepsilon,\delta,n,2^j\Big)$ of $\mathrm{BaseSumDP}$
}

Send $\big\{S_{i}^{[2^{j-1}+1,2^j]}\}_{j\in\{ 0,1,2,\dots,\log U\}}$\;
\caption{Randomizer of $\mathrm{SumDP}$.}
\label{alg:sum_randomizer}
\end{algorithm*}

\begin{algorithm*}[t]
\LinesNumbered 
\SetNoFillComment
\KwIn{$\Big\{R^{[2^{j-1}+1,2^j]}=\cup_i S_{i}^{[2^{j-1}+1,2^j]}\Big\}_{j\in\{ 0,1,2,\dots,\log U\}}$ with $S_{i}^{[2^{j-1}+1,2^j]}$ from user $i$, $\varepsilon$, $\delta$, $\beta$, $n$, $U$}
$\tau\gets 0$\;
\For{$j\gets 0,1,2,\dots,\log U$}{
    $\widetilde{\mathrm{Sum}}\big(D\cap[2^{j-1}+1,2^j]\big)\gets \sum_{y\in R^{[2^{j-1}+1,2^j]}}y$\;
    \tcc{Set $\tau=2^j$ for last $j$ passing the condition of (\ref{eq:condition})}
    \If{$\widetilde{\mathrm{Sum}}\big(D\cap[2^{j-1}+1,2^j]\big)>1.3\cdot 2^j\cdot \ln\big(2(\log U+1)/\beta\big)/\varepsilon$}{
        $\tau\gets 2^j$\;
    }
}
$\widetilde{\mathrm{Sum}}(D)\gets \sum_{j\in \{0,1,2,\dots,\log \tau\}}\widetilde{\mathrm{Sum}}\big(D\cap[2^{j-1}+1,2^j]\big)$\;
\Return $\widetilde{\mathrm{Sum}}(D)$
\caption{Analyzer of $\mathrm{SumDP}$.}
\label{alg:sum_analyzer}
\end{algorithm*}

\subsection{Finding $\tau$ with No Extra Cost}
It remains to deal with the impractical $\tau$-selection protocol used in \cite{huang2021instance}, which has a message complexity of $O((\log\log U)^6 \log(1/\delta)/\varepsilon)$ and finds a $\tau$ such that
\begin{equation}
\label{eq:tau}
\mathrm{Max}(D,k)\leq \tau \leq 2\cdot {\mathrm{Max}}(D),
\end{equation}
for some $k = \tilde{O}\big((\log\log U)^{3.5}\sqrt{\log(1/\delta)}/\varepsilon\big)$.  Recall that the bias introduced by the clipping is $O(\Max(D)\cdot k)$. 

It turns out that we can find a $\tau$ that achieves the optimal $k=O(\log\log U/\varepsilon)$ with no extra cost at all!  
To illustrate the idea, first consider the non-private setting where we have access to the exact values of $\mathrm{Sum}\big(D\cap[2^{j-1}+1,2^j]\big)$ for each $j=0,1,2,\dots, \log U$.  Then it is easy to see that the last $j$ on which $\mathrm{Sum}\big(D\cap[2^{j-1}+1,2^j]\big) > 0$ yields a $\tau=2^j$ such that $\Max(D) \le \tau \le 2\cdot \Max(D)$, i.e., we can achieve \eqref{eq:tau} with $k=1$. In the private setting, however, due to having access only to the noisy estimates $\widetilde{\mathrm{Sum}}\big(D\cap[2^{j-1}+1,2^j]\big)$, we may easily overshoot: With probability at least $1/2$, the last sub-domain has $\widetilde{\mathrm{Sum}}\big(D\cap[\tau/2-1,\tau]\big)>0$ (even if it is empty), which would set $\tau=U$. 

To prevent this overshooting, our idea is to use a higher bar.  Instead of finding the last $j$ on which $\mathrm{Sum}\big(D\cap[2^{j-1}+1,2^j]\big) > 0$, we change the condition to 
\begin{equation}
    \label{eq:condition}
\widetilde{\mathrm{Sum}}\big(D\cap[2^{j-1}+1,2^j]\big)> 1.3\cdot 2^j\cdot \ln\big(2(\log U+1)/\beta\big)/\varepsilon. 
\end{equation}
The RHS of \eqref{eq:condition} follows from the error bound of $\mathrm{BaseSumDP}$ (see the remark after Lemma \ref{lm:baseline}), where we replace $U$ with $2^j$ (since the largest value in this sub-domain is $2^j$) and replace $\beta$ with $\beta/(\log U +1)$, so that when this sub-domain is empty, \eqref{eq:condition} happens with probability at most $\beta/(\log U +1)$.  Then by a union bound, with probability at least $1-\beta$, none of the empty sub-domains passes the condition \eqref{eq:condition}.  In this case, we are guaranteed to find a $\tau=2^j \le 2 \cdot \Max(D)$, namely, we will not overshoot.  Meanwhile, we can also show that
we will not undershoot too much, either.  More precisely, with probability at least $1-\beta$, there are at most $O(\log(\log U/\beta)/\varepsilon)$ elements greater than $2^j$.
Therefore, plugging $\tau=2^j$ into the clipping mechanism yields the optimal central-DP error of $O(\Max(D)\cdot \log(\log U/\beta) / \varepsilon)$.


To summarize, our final protocol works as follows.  After domain partitioning, 
each user $i$ executes an instance of $\mathrm{BaseSumDP}$ for every sub-domain $[2^{j-1}+1,2^j]$ with the input $x_i\cdot \mathbf{I}\big(x_i\in [2^{j-1}+1,2^j]\big)$ and the whole privacy budget $\varepsilon$, $\delta$.  As all the messages are shuffled together, they need to identify themselves with which $\mathrm{BaseSumDP}$ instance they belong to.  This just requires extra $O(\log\log U)$ bits.
From the perspective of the analyzer, based on the received messages, we compute $\widetilde{\mathrm{Sum}}\big(D\cap [2^{j-1}+1,2^j]\big)$ for each $j\in\{0,1,\dots,\log U\}$. 
Then, $\tau$ is set to $2^j$ for the last $j$ passing condition \eqref{eq:condition}.
Finally, we sum up all $\widehat{\mathrm{Sum}}\big(D\cap [2^{j-1}+1,2^j]\big)$ for $j\leq \log\tau$.
The detailed algorithms for the randomizer and analyzer are presented in Algorithm~\ref{alg:sum_randomizer} and Algorithm~\ref{alg:sum_analyzer}.
Besides, we give an example to demonstrate the protocol in Figure~\ref{fig:illustration}.

\begin{restatable}{theorem}{thsum}
Given any $\varepsilon>0$, $\delta>0$, $n\in\mathbb{Z}_+$, and $U\in\mathbb{Z}_+$, for any $D\in [U]^{n}$,  $\mathrm{SumDP}$ achieves the following:
\begin{enumerate}
    \item The messages received by the analyzer preserves $(\varepsilon,\delta)$-DP; 
    \item With probability at least $1-\beta$,
    the error is bounded by 
    \[O\Big(\mathrm{Max}(D)\cdot\log(\log U/\beta)/\varepsilon\Big);\]
    \item In expectation, each user sends  $1+ O\big(\log U\cdot\log^2 n\cdot \log(1/\delta)/(\varepsilon\sqrt{n})\big)$ messages with each containing $O\Big(\min\big(\log n,\log U\big)\Big)$ bits. 
\end{enumerate}
\label{th:sum}
\end{restatable}

\begin{proof}
For privacy, invoking Lemma~\ref{lm:baseline}, we have that for every $j\in\{0,1,2,\dots,\log U\}$, $R^{[2^{j-1}+1,2^j]}$ preserves $(\varepsilon,\delta)$-DP. Given that each $x_i$ only has an impact on a single $R^{[2^{j-1}+1,2^j]}$, it follows that the collection $\big\{R^{[2^{j-1}+1,2^j]}\big\}_{j\in{ 0,1,2,\dots,\log U}}$ preserves $(\varepsilon,\delta)$-DP.

For utility, Lemma~\ref{lm:baseline} ensures that for any $j\in \{0,1,2,\dots,\log U\}$, with probability at least $1-\frac{\beta}{\log U+2}$, we have
\begin{align}
\nonumber
&\Big|\widetilde{\mathrm{Sum}}\big(D\cap[2^{j-1}+1,2^j]\big)-{\mathrm{Sum}}\big(D\cap[2^{j-1}+1,2^j]\big)\Big|
\\
\leq & 1.3\cdot 2^j\cdot \ln\big(2(\log U+2)/\beta\big)/\varepsilon.
\label{eq:th:sum_1}
\end{align}
Aggregating probabilities across all $j$ yields that, with probability at least $1-\beta$, (\ref{eq:th:sum_1}) holds across all $j$. 

First, (\ref{eq:th:sum_1}) implies $\tau$ will not surpass $\Big\lceil\log\big(\mathrm{Max}(D)\big)\Big\rceil$:
\begin{equation*}
\tau \leq \Big\lceil\log\big(\mathrm{Max}(D)\big)\Big\rceil.
\end{equation*}
Combining this with (\ref{eq:th:sum_1}), we have
{
\small
\begin{align}
\nonumber
&\Big|\sum_{j=0}^{\log(\tau)}\big(\widetilde{\mathrm{Sum}}\big(D\cap[2^{j-1}+1,2^j]\big)-{\mathrm{Sum}}\big(D\cap[2^{j-1}+1,2^j]\big)\big)\Big|
\\
\label{eq:th:sum_2}
= & O \Big(\mathrm{Max}(D)\cdot \log(\log U/\beta)/\varepsilon\Big).
\end{align}
}

Meanwhile, with (\ref{eq:th:sum_1}), we also have that all sub-domains over $\tau$ will not contain too many elements: for any $j> \log(\tau)$,
\begin{equation*}
{\mathrm{Sum}}\big(D\cap[2^{j-1}+1,2^j]\big)\leq 2.6\cdot 2^j\cdot \ln\big(2(\log U+2)/\beta\big)/\varepsilon,
\end{equation*}
which sequentially deduces
\begin{align}
\nonumber
&\sum_{j=\log\tau+1}^{\log U} {\mathrm{Sum}}\big(D\cap[2^{j-1}+1,2^j]\big)  
\\
= & O \Big(\mathrm{Max}(D)\cdot \log(\log U/\beta)/\varepsilon\Big).
\label{eq:th:sum_3}
\end{align}

Finally, (\ref{eq:th:sum_2}), and (\ref{eq:th:sum_3}) lead to our desired statement for the utility.

The statement for communication directly follows from Lemma~\ref{lm:baseline} and the observation that each $x_i$ uniquely corresponds to a single interval $[2^{j-1}+1,2^j]$.
\end{proof}

Additionally, each randomizer incurs a computational cost of $O\big(\log U \cdot \min(U,\sqrt{n})\big)$ and each analyzer operates with a running time of $O(n)$.

\section{High-Dimensional Sum Estimation}
\label{sec:highD}
In this section, we consider the high-dimensional scenario, i.e., each $x_i$ is a $d$-dimensional vector in $\mathbb{Z}^d$ with $\ell_2$ norm bounded by some given (potentially large) $U_{\ell_2}$.  Thus, $D$ can also be thought of as an $n\times d$ matrix. Let $\mathrm{Max}_{\ell_2}(D) := \max_{x_i \in D}\|x_i\|_2$ be the maximum $\ell_2$ norm among the elements (columns) of $D$.  The goal is to estimate $\mathrm{Sum}(D) = \sum_{i}x_i$ with small $\ell_2$ error.  For each $x_i\in \mathbb{Z}^d$ and any $k\in [d]$, we use $x_i^k$ to denote its $k$-th coordinate.

In the central model, the standard \textit{Gaussian mechanism} achieves an error of 
$O\big(U_{\ell_2}\sqrt{d\log(1/\delta)}/\varepsilon\big)$, which is worst-case optimal up to logarithmic factors \cite{kamath2020primer}.   The best clipping mechanism for this problem \cite{dong2023better}  achieves an error of \[O\Big(\mathrm{Max}_{\ell_2}(D)\cdot \big(\sqrt{d\log(1/\delta)}+\log\log(U_{\ell_2})\big)/\varepsilon\Big).\] 

In the shuffle model, 
\cite{huang2021instance} presented a two-round protocol achieving a (theoretically) similar bound:
\begin{align*}
\tilde{O}\bigg(\mathrm{Max}_{\ell_2}(D)\cdot \Big(\sqrt{d\log(nd)\log(1/\delta)}
\\
+\log^{3.5}U_{\ell_2}\cdot\sqrt{\log(1/\delta)}\Big)/\varepsilon\bigg).  
\end{align*}
But similar to their 1D protocol, this algorithm is not practical due to the $\log^{3.5}U_{\ell_2}$ factor and the use of the complicated range-counting shuffle-DP protocol of \cite{ghazi2021power}.

In this section, we present a simple and practical one-round shuffle-DP protocol that achieves an error of 
\[O\Big(\mathrm{Max}_{\ell_2}(D)\cdot\sqrt{d\log(nd)\log(1/\delta)} \cdot\log\log(U_{\ell_2})/\varepsilon\Big).\]
%

\begin{figure*}[t]
\centering
\includegraphics[width=1\textwidth]{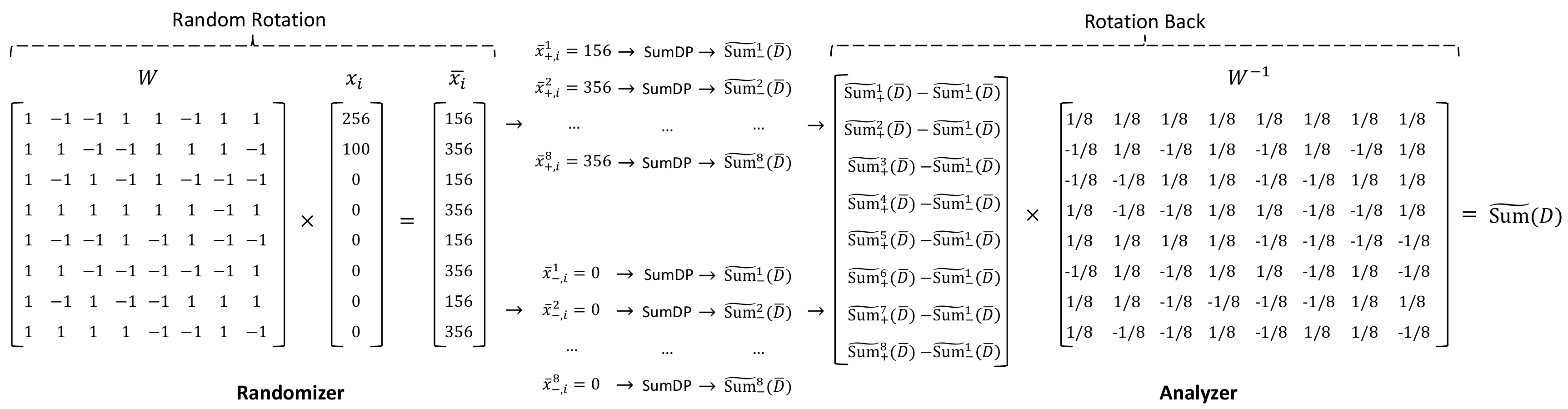}
    \caption{An illustration of our protocol for high-dimensional sum estimation. $d=8$.}
    \label{fig:illustration_high_dim}
\end{figure*}

\begin{algorithm*}[t]
\LinesNumbered 
\KwIn{$x_i\in \mathbb{Z}_{\geq 0}^d$, $\varepsilon$, $\delta$, $\beta$, $n$, $U_{\ell_2}$, $W$}
$\Bar{x}_i\gets W x_i$\;
$\Bar{x}_{+,i}\gets \big(\Bar{x}_{+,i}^{1},\Bar{x}_{+,i}^{2},\dots,\Bar{x}_{+,i}^{d}\big)$ with $\Bar{x}_{+,i}^{k}\gets \min\Big(\Bar{x}_{i}^{1}\cdot \mathbf{I}(\Bar{x}_{i}^{1}>0),U_{\ell_2}\sqrt{2\log(8nd/\beta)}\Big)$ for any $k\in[d]$\;
$\Bar{x}_{-,i}\gets \big(\Bar{x}_{-,i}^{1},\Bar{x}_{-,i}^{2},\dots,\Bar{x}_{-,i}^{d}\big)$ with $\Bar{x}_{-,i}^{k}\gets \min\Big(-\Bar{x}_{i}^{1}\cdot \mathbf{I}(\Bar{x}_{i}^{1}<0),U_{\ell_2}\sqrt{2\log(8nd/\beta)}\Big)$ for any $k\in[d]$\;
$\varepsilon',\delta'\gets \varepsilon/\big(4\sqrt{d\log(2/\delta)}\big),\delta/(4d)$\;
\For{$k\gets 1,2,\dots,d$}{
    $S_{+,i}^{k}\gets \mathrm{Randomizer}\Big(\bar{x}^k_{+,i},\varepsilon',\delta',n,\beta/(2d),U_{\ell_2}\sqrt{2\log(8nd/\beta)}\Big)$ of $\mathrm{SumDP}$\;
    $S_{-,i}^{k}\gets \mathrm{Randomizer}\Big(\bar{x}^k_{-,i},\varepsilon',\delta',n,\beta/(2d),U_{\ell_2}\sqrt{2\log(8nd/\beta)}\Big)$ of $\mathrm{SumDP}$\;
}
Send $\{S_{-,i}^{k}\}_{k\in[d]}$, $\{S_{+,i}^{k}\}_{k\in[d]}$\;
\caption{Randomizer of $\mathrm{HighDimSumDP}$.}
\label{alg:high_dim_sum_randomizer}
\end{algorithm*}

\begin{algorithm*}[t]
\LinesNumbered 
\KwIn{$\{R^k_+=\cup_i S_{+,i}^{k}\}_{k\in[d]}$ and $\{R^k_-=\cup_i S_{-,i}^{k}\}_{k\in[d]}$ with $\{S_{+,i}^{k}\}_{k\in[d]}$ and $\{S_{-,i}^{k}\}_{k\in[d]}$ from user $i$, $\varepsilon$, $\delta$, $\beta$, $n$, $U_{\ell_2}$, $W$}
$\varepsilon',\delta'\gets \varepsilon/\big(4\sqrt{d\log(2/\delta)}\big),\delta/(4d)$\;
\For{$k\gets 1,2,\dots,d$}{
    $\widetilde{\mathrm{Sum}}_+^k(\overline{D})\gets \mathrm{Analyzer}\Big(R^k_+,\varepsilon',\delta',\beta/(2d),U_{\ell_2}\sqrt{2\log(8nd/\beta)}\Big)$ of $\mathrm{SumDP}$\;
    $\widetilde{\mathrm{Sum}}_-^k(\overline{D})\gets \mathrm{Analyzer}\Big(R^k_-,\varepsilon',\delta',\beta/(2d),U_{\ell_2}\sqrt{2\log(8nd/\beta)}\Big)$ of $\mathrm{SumDP}$\;
}
$\widetilde{\mathrm{Sum}}(\overline{D})\gets \Big(\widetilde{\mathrm{Sum}}_+^1(\overline{D})-\widetilde{\mathrm{Sum}}_-^1(\overline{D}),\widetilde{\mathrm{Sum}}_+^2(\overline{D})-\widetilde{\mathrm{Sum}}_-^2(\overline{D}),\dots, \widetilde{\mathrm{Sum}}_+^d(\overline{D})-\widetilde{\mathrm{Sum}}_-^d(\overline{D})\Big)$\;
$\widetilde{\mathrm{Sum}}(D)\gets W^{-1} \widetilde{\mathrm{Sum}}(\overline{D})$\;
\Return $\widetilde{\mathrm{Sum}}(D)$
\caption{Analyzer of $\mathrm{HighDimSumDP}$.}
\label{alg:high_dim_sum_analyzer}
\end{algorithm*}

\subsection{Random Rotation}
As in \cite{huang2021instance}, we first perform a random rotation of the dataset $D$, resulting in $\overline{D} = WD$. Here, $W$ denotes a rotation matrix, constructed as per the following lemma:

\begin{lemma}[\cite{ailon2009fast}]
\label{lm:rotation}
Let $W = HP$, where $H$ is the Hadamard matrix and $P$ is a diagonal matrix whose diagonal entry is independently and randomly sampled from $\{-1,+1\}$.  Then, for any $x\in \mathbb{Z}_{\geq 0}^d$, and any $\beta>0$, we have
\begin{enumerate}
    \item $Wx \in \mathbb{Z}^d$ and $\|Wx\|_2 = \sqrt{d} \|x\|_2$\;
    \item 
    \[\mathsf{Pr}\Big[\big\|Wx\big\|_{\infty}\geq \|x\|_2\cdot \sqrt{2\log(4d/\beta)}\Big]\leq \beta.\]
\end{enumerate}
\end{lemma}
The first property means, matrix $W$ performs a rotation while preserving the integer domain, aside from scaling the vector by a factor of $\sqrt{d}$.  The second property ensures that the random rotation spreads out the norm evenly across all dimensions.
After the rotation, \cite{huang2021instance} clips each coordinate to $O(U_{\ell_2}\log(nd))$ and invokes $\mathrm{BaseSumDP}$ with bounded domain size $O(U_{\ell_2}\log(nd))$ in each dimension.
To guarantee DP, they use advanced composition to allocate each dimension with the privacy budget $\varepsilon' = \varepsilon/\big(2\sqrt{d\log(2/\delta)}\big)$, $\delta' = \delta/(2d)$.
This approach results in an error of $\tilde{O}\Big( U_{\ell_2} d\sqrt{\log n\log(1/\delta)}/\varepsilon \Big)$ for the estimation of $\mathrm{Sum}\big(\overline{D}\big)$.  
Upon reorienting to the original domain with $W^{-1}$, the estimation of $\mathrm{Sum}(D)$ has error $\tilde{O}\Big( U_{\ell_2} \sqrt{d\log n\log(1/\delta)}/\varepsilon \Big)$.  
Note that the randomness in $W$ is only needed for the utility analysis, and does not affect privacy, so it can be derived from public randomness.

\subsection{Extending $\mathrm{SumDP}$ to High Dimensions}

To adapt $\mathrm{SumDP}$ to high dimensions, one na\"ive approach is to use the advanced composition to divide the privacy budget and apply $\mathrm{SumDP}$ to each dimension.  This will lead to an error of 
\begin{equation}
\label{eq:naive_high_dim_error}
\Tilde{O}\left(\sqrt{d\log(1/\delta)\cdot\log\log U_{\ell_2}\cdot\sum_{k=1}^d\mathrm{Max}(D^{(k)})}\right),
\end{equation}
where $\mathrm{Max}(D^{(k)})$ is the maximum value in the $k$ dimension of $D$.  Since $\sqrt{\sum_{k=1}^d\mathrm{Max}(D^{(k)})}$ can be as large as $\sqrt{d}\cdot \mathrm{Max}_{\ell_2}(D)$, (\ref{eq:naive_high_dim_error}) has a $\sqrt{d}$ degradation compared with the optimal error.

To achieve our error with a dependency on $\sqrt{d}$,
we apply a rotation as in Lemma~\ref{lm:rotation}, and then clip each coordinate to the range $[-c\cdot U_{\ell_2}\log(nd),c\cdot U_{\ell_2}\log(nd)]$ for some constant $c$. 
One would then apply $\mathrm{SumDP}$ in each dimension. However, after the rotation, the resulting domain of $\overline{D} = HD$ spans both positive and negative integers.  Note that $\mathrm{SumDP}$ only works on the non-negative integer domain, because its utility guarantee is based on the property that, clipping elements should only make the sum smaller, which is not true if negative numbers are present. 

A simplistic strategy is to shift the domain from $[-c\cdot U_{\ell_2}\log(nd),c\cdot U_{\ell_2}\log(nd)]$ to $[0,2c\cdot U_{\ell_2}\log(nd)]$, and then apply $\mathrm{SumDP}$. However, this shift escalates the maximum value for each coordinate to $\ge c\cdot U_{\ell_2}\log(nd)$, potentially inducing an error proportional to $U_{\ell_2}$ in the estimation of $\mathrm{Sum}(\overline{D})$.
To fix this issue, we process the positive and negative domains separately, and then take their difference as the final estimate for $\mathrm{Sum}(\overline{D})$. 
We call this algorithm $\mathrm{HighDimSumDP}$.
The detailed algorithms for the randomizer and analyzer are shown in Algorithm~\ref{alg:high_dim_sum_randomizer} and Algorithm~\ref{alg:high_dim_sum_analyzer}.
We also show an example in Figure~\ref{fig:illustration_high_dim}.

\begin{restatable}{theorem}{thhighdimsum}
\label{th:high_dim_sum}
Given any $\varepsilon>0$, $\delta>0$, $n\in\mathbb{Z}_+$, and $U\in\mathbb{Z}_+$, for any $D\in \mathbb{Z}^{n\times d}$ with $\mathrm{Max}_{\ell_2}(D)\leq U$, we have
\begin{enumerate}
    \item The messages received by the analyzer preserves $(\varepsilon,\delta)$-DP; 
    \item With probability at least $1-\beta$,
    the $\ell_2$ error is bounded by 
\begin{align*}
O\Big(\mathrm{Max}_{\ell_2}(D)\cdot\sqrt{d\log(nd/\beta)\log(1/\delta)} \\\cdot\log(d\log(U_{\ell_2})/\beta)/\varepsilon\Big);
\end{align*}
    \item In expectation, each user sends 
    \begin{align*}
    d+ O\big(d^{1.5}\cdot\log(U_{\ell_2}\log(nd/\beta))\\\cdot \log^{1.5}(d/\delta)\cdot\log^2 n/(\varepsilon\sqrt{n})\big)
    \end{align*}
 messages with each message containing $O\Big(\log d+\min\big(\log n,\log(U_{\ell_2})\big)$ bits. 
\end{enumerate}
\end{restatable}

\begin{proof} 
For privacy, invoking Theorem~\ref{th:sum}, we deduce that each of $R_+^k$ and $R_-^k$ adheres to $\Big(\frac{\varepsilon}{4\sqrt{d\log(2/\delta)}},\frac{\delta}{4d}\Big)$-DP. By advanced composition, the collections of $\{R_+^k\}{k\in[d]}$ and $\{R-^k\}_{k\in[d]}$ maintain $(\varepsilon,\delta)$-DP.

Regarding utility, Lemma~\ref{lm:rotation} guarantees, with probability at least $1-\frac{2}{\beta}$, for every $i\in[n]$ and $j\in[d]$, we have
\begin{equation}
|\Bar{x}_i^{k}| \leq \mathrm{Max}_{\ell_2}(D) \cdot \sqrt{2\log(8nd/\beta)},
\label{eq:th:high_dim_sum_1}
\end{equation}
which further implies
\begin{equation}
\label{eq:th:high_dim_sum_2}
\sum_{i}\Bar{x}_i = \sum_i\Bar{x}_{+,i}-\sum_i\Bar{x}_{-,i}.
\end{equation}

Subsequently, Theorem~\ref{th:sum} coupled with Equation~\ref{eq:th:high_dim_sum_1} implies that for each $k\in[d]$, with probability at least $1-\frac{\beta}{2d}$,
\begin{align}
\nonumber
\Big|\widetilde{\mathrm{Sum}}_+^k(\overline{D})&-\sum_i\Bar{x}_{+,i}^k\Big| = O\Big(\mathrm{Max}_{\ell_2}(D)\cdot\sqrt{d\log(1/\delta)}
\\
\label{eq:th:high_dim_sum_3}
& \cdot \sqrt{\log(nd/\beta)}\cdot\log(d\log(U_{\ell_2})/\beta)/\varepsilon\Big) 
\end{align}
and
\begin{align}
\nonumber
\Big|\widetilde{\mathrm{Sum}}_-^k(\overline{D})&-\sum_i\Bar{x}_{-,i}^k\Big| = O\Big(\mathrm{Max}_{\ell_2}(D)\cdot\sqrt{d\log(1/\delta)}
\\
\label{eq:th:high_dim_sum_4}
& \cdot \sqrt{\log(nd/\beta)}\cdot\log(d\log(U_{\ell_2})/\beta)/\varepsilon\Big) 
\end{align}
Combining the probabilities across all $k\in[d]$, we have with probability at least $1-\frac{\beta}{2}$, (\ref{eq:th:high_dim_sum_3}) and (\ref{eq:th:high_dim_sum_4}) hold for all dimensions.  

By synthesizing (\ref{eq:th:high_dim_sum_2}), (\ref{eq:th:high_dim_sum_3}), and (\ref{eq:th:high_dim_sum_4}), we
\begin{align*}
&\Big\|\widetilde{\mathrm{Sum}}(\overline{D})-\mathrm{Sum}(\overline{D})\Big\|_2
\\
=&\frac{1}{\sqrt{d}}\cdot \Big\|\widetilde{\mathrm{Sum}}(\overline{D})-\mathrm{Sum}(\overline{D})\Big\|_2 
\\
=& O\Big(\mathrm{Max}_{\ell_2}(D)\cdot\sqrt{d\log(nd/\beta)\log(1/\delta)}
\\
&\cdot\log(d\log(U_{\ell_2})/\beta)/\varepsilon\Big) 
\end{align*}

Finally, our assertion on communication cost derives from (\ref{eq:th:high_dim_sum_1}), Theorem~\ref{th:sum}, along with the observation that for any $i\in[n]$ and $k\in[d]$, either $\Bar{x}_{+,i}^k$ or $\Bar{x}_{-,i}^k$ is necessarily zero.
\end{proof}

\begin{algorithm*}[t]
\LinesNumbered 
\SetNoFillComment
\KwIn{$x_i\in \{0,1\}^d$, $\varepsilon$, $\delta$, $\beta$, $n$}
\For{$j\gets 0,1,2,\dots,\log d$}{
    \tcc{The messages for counting vectors with sparsity between $2^{j-1}+1$ and $2^j$}
    $S_{\mathrm{cnt},i}^{[2^{j-1}+1,2^j]}\gets\mathrm{Randomizer}\Big(\mathbf{I}\big(|x_i|_1\in [2^{j-1}+1,2^j]\big),\varepsilon/2,\delta/2,n,1\Big)$ of $\mathrm{BaseSumDP}$\;
    \tcc{The messages for sum for vectors with sparsity between $2^{j-1}+1$ and $2^j$}
    $\varepsilon',\delta'\gets \varepsilon/\big(2\sqrt{ 2^{j+1}\log(2/\delta)}\big),\delta/(2^{j+1})$\;
    \For{$k\gets 1,2,\dots,d$}{
        $S_{\mathrm{sum},i}^{[2^{j-1}+1,2^j],k}\gets\text{Randomizer}\big(x_i^k\cdot\mathbf{I}(|x_i|_1\in [2^{j-1}+1,2^j]),\varepsilon',\delta',n,1\big)$ of $\mathrm{BaseSumDP}$\;
    }
}
Send $\big\{S_{\mathrm{sum},i}^{[2^{j-1}+1,2^j],k}\}_{j\in\{ 0,1,2,\dots,\log d\},k\in[d]}$ and $\big\{S_{\mathrm{cnt},i}^{[2^{j-1}+1,2^j]}\}_{j\in\{ 0,1,2,\dots,\log d\}}$\;
\caption{Randomizer of $\mathrm{SparVecSumDP}$.}
\label{alg:sparse_vector_randomizer}
\end{algorithm*}

\begin{algorithm*}[t]
\LinesNumbered 
\KwIn{$\big\{R_{\mathrm{sum}}^{[2^{j-1}+1,2^j],k}\}_{j\in\{ 0,1,2,\dots,\log d\},k\in[d]}= \cup_i S_{\mathrm{sum},i}^{[2^{j-1}+1,2^j],k}\}_{j\in\{ 0,1,2,\dots,\log d\},k\in[d]}$ and $\big\{R_{\mathrm{cnt}}^{[2^{j-1}+1,2^j]}=\cup_i S_{\mathrm{cnt},i}^{[2^{j-1}+1,2^j]}\}_{j\in\{ 0,1,2,\dots,\log d\}}$, $\varepsilon$, $\delta$, $\beta$, $n$}
$\tau\gets 0$\;
\For{$j\gets 0,1,2,\dots,\log d$}{
    $\widetilde{\mathrm{Count}}\big(D[2^{j-1}+1,2^j]\big)\gets \mathrm{Analyzer}\Big(R_{\mathrm{cnt}}^{[2^{j-1}+1,2^j]},\varepsilon/2,\delta/2,n,1\Big)$ of $\mathrm{BaseSumDP}$\;
    $\varepsilon',\delta'\gets \varepsilon/\big(2\sqrt{ 2^{j+1}\log(2/\delta)}\big),\delta/(2^{j+1})$\;
    \For{$k\gets 1,2,\dots,d$}{
        $\widetilde{\mathrm{Sum}}^k\big(D[2^{j-1}+1,2^j]\big)\gets \mathrm{Analyzer}\Big(R_{\mathrm{sum}}^{[2^{j-1}+1,2^j],k},\varepsilon',\delta',n,1\Big)$ of $\mathrm{BaseSumDP}$\;
    }
    $\widetilde{\mathrm{Sum}}\big(D[2^{j-1}+1,2^j]\big)\gets\Big(\widetilde{\mathrm{Sum}}^1\big(D[2^{j-1}+1,2^j]\big),\widetilde{\mathrm{Sum}}^2\big(D[2^{j-1}+1,2^j]\big),\dots,\widetilde{\mathrm{Sum}}^d\big(D[2^{j-1}+1,2^j]\big)\Big)$\;
    \If{$\widetilde{\mathrm{Count}}\big(D[2^{j-1}+1,2^j]\big) > 1.3\cdot \frac{2}{\varepsilon}\cdot\log(2(\log d +1)/\beta)$}{
        $\tau\gets 2^j$\;
    }
}

$\widetilde{\mathrm{Sum}}(D)\gets\sum_{j\in\{0,1,2,\dots,\log \tau\}}\widetilde{\mathrm{Sum}}\big(D[2^{j-1}+1,2^j]\big)$\;
\caption{Analyzer of $\mathrm{SparVecSumDP}$.}
\label{alg:sparse_vector_analyzer}
\end{algorithm*}

\section{Sparse Vector Aggregation}
\label{sec:sparse}

As the last application of our technique, we study the sparse vector aggregation problem.  In this problem, each $x_i$ is a binary vector in $\{0,1\}^d$. We use $(\mathrm{Max}_{\ell_2}(D))^2 =\max_i \|x_i\|_1$ to quantify the data's sparsity and are interested in the sparse case where $(\mathrm{Max}_{\ell_2}(D))^2\ll d$.
We want to estimate $\mathrm{Sum}(D)$ with an $\ell_\infty$ error $\Max_{\ell_2}(D)/\varepsilon \cdot \mathrm{poly}\log(d/\delta)$.  Meanwhile, we would like the message complexity of each user $i$ to depend on $\|x_i\|_1$, i.e., the number of $1$'s in $x_i$.
Note that the $\ell_2$ error in Theorem~\ref{th:high_dim_sum} can only imply the same $\ell_\infty$ error, namely, it is  $\sqrt{d}$ times larger than desired. Moreover, it requires $d$ messages per user.

\subsection{Clipping on Sparsity }
\label{sec:clipping_sparsity}

If an upper bound of sparsity $S\geq \big(\mathrm{Max}_{\ell_2}(D)\big)^2$ is given,  we can estimate the count for each coordinate independently with the privacy budget $\varepsilon' = \varepsilon/\big(\sqrt{S\log(2/\delta)}\big)$, $\delta' = \delta/(2S)$. Given that each $x_i$ at most affects the counting for $S$ dimensions, with advanced composition, this whole process preserves $(\varepsilon,\delta)$-DP. 
The state-of-the-art protocol for counting under the shuffle-DP model is $\mathrm{BaseSumDP}$ without random rounding, where the communication is improved to $1+(\log(1/\delta)/(\varepsilon n))$ messages per user.  Feeding this into the above protocol for sparse vector aggregation yields an error proportional to $\sqrt{S}$ and $\|x_i\|_1+O(d^{1.5}\log^{1.5}/(\varepsilon n))$ messages per user.  

In the absence of a good upper bound $S$, one could apply the clipping mechanism on sparsity.  Specifically, for some $\tau$, we only retain the first $\tau$ non-zero coordinates of each $x_i$ and set the rest to $0$.  Then we apply the mechanism above with $S=\tau$.  However, as in the sum estimation problem, the key is to choose a good $\tau$ that balances the DP noise and bias, and the optimal $\tau$ should achieve an error proportional to $\Max_{\ell_2}(D)$.  More importantly, we would like to choose $\tau$ and compute the noisy counts of all dimensions clipped by $\tau$ simultaneously in one round. 

\subsection{Sparsity Partitioning}

We use the idea of domain partitioning from our sum estimation protocol.  But for the sparse vector aggregation problem, we partition the domain of possible sparsity levels $[d]$ into $\log d+1$ disjoint sub-domains: $[1,1]$, $[2,2]$, $[3,4]$, $\dots$, $[d/2+1,d]$.  Then, we divide the vectors according to their sparsity. More precisely, for each $j\in\{0,1,2,\dots,\log d\}$, let 
\[D[2^{j-1}+1,2^j] = \big\{x_i\in D:\|x_i\|_1\in [2^{j-1}+1,2^j] \big\}.\]
Since each vector in $D[2^{j-1}+1,2^j]$ has the sparsity bounded by  $2^j$, we can use the idea discussed in the last section.  For the error, the estimation of $\mathrm{Sum}\big(D[2^{j-1}+1,2^j]\big)$ has an $\ell_{\infty}$ error bounded by $\Tilde{O}(\sqrt{2^j})$.  In terms of the communication, since each $x_i$ will only be involved in $D[2^{j-1}+1,2^j]$, each user sends $\|x_i\|_1+\Tilde{O}(d^{1.5}\log^{1.5}(1/\delta)/(\varepsilon n))$ messages in expectation.

Next, let us discuss how to use the estimations of $\mathrm{Sum}(D[2^{j-1}+1,2^j])$ to reconstruct $\mathrm{Sum}(D)$.  Recall that, in sum estimation, we have the estimations for each value domain, i.e., $\mathrm{Sum}(D\cap[2^{j-1}+1,2^j])$ find the last $[2^{j-1}+1,2^j]$ with a large noisy sum result.
This is to guarantee enough elements are located in the domain $[2^{j-1}+1,2^j]$.
Unfortunately, such an idea cannot be extended to the high-dimensional case directly.   
The problem is that, even though, there are a large number of vectors with the sparsity in the range of $[2^{j-1}+1,2^j]$, i.e., $\big|D[2^{j-1}+1,2^j]\big|$ is large enough, each coordinate of $\mathrm{Sum}\big(D[2^{j-1}+1,2^j]\big)$ can still be very small since those vectors can contribute totally different coordinates.

The solution here is that we build an extra counter for the number of vectors with sparsity within each $[2^{j-1}+1,2^j]$ as a judgment for whether to include $\mathrm{Sum}\big(D[2^{j-1}+1,2^j]\big)$ in the final result. 
More precisely, each user first executes $\log d+1$ number of instances of $\mathrm{BaseSumDP}$, each of which is to estimate the number of vectors with sparsity within $[2^{j-1}+1,2^j]$ and uses privacy budget $\varepsilon/2$ and $\delta/2$.  
Then, for each $j\in \{0,1,2,\dots,\log d\}$, we estimate $\mathrm{Sum}\big(D[2^{j-1}+1,2^j]\big)$, where we use one $\mathrm{CounDP}$ with the privacy budget $\varepsilon/\big(2\sqrt{2^{j+1}\log(2/\delta)}\big)$ and $\delta/(2^{j+1})$ to do the sum estimation in $k$th coordinate.  
In the view of the analyzer, with the received messages, we can easily get the estimation for $\mathrm{Count}\big(D[2^{j-1}+1,2^j]\big)$ and $\mathrm{Sum}\big(D[2^{j-1}+1,2^j]\big)$ for each $j\in\{0,1,2,\dots,\log d\}$.
We set $\tau = 2^j$ with the last $j$ such that $\mathrm{Count}\big(D[2^{j-1}+1,2^j]\big)$ is large enough.  Finally, we sum all estimations for $\mathrm{Sum}\big(D[2^{j-1}+1,2^j]\big)$ for $j\leq \log(\tau)$.
The detailed algorithms for the randomizer and analyzer are shown in Algorithms~\ref{alg:sparse_vector_randomizer} and \ref{alg:sparse_vector_analyzer}.

\begin{restatable}{theorem}{thsparsevector}
Given any $\varepsilon>0$, $\delta>0$, $n\in\mathbb{Z}_+$, and for any $D\in \{0,1\}^{n\times d}$, the $\mathrm{SparsVecSumDP}$ achieves the following:
\begin{enumerate}
    \item The messages received by the analyzer preserves $(\varepsilon,\delta)$-DP; 
    \item With probability at least $1-\beta$, for every $k\in [d]$, the $\ell_{\infty}$ error is bounded by 
    \[O\Big(\big(\mathrm{Max}_{\ell_2}(D)\cdot\sqrt{\log(1/\delta)}+\log\log d\big)\cdot\log(d/\beta)/\varepsilon\Big);\]
    \item In expectation, each user sends  $\|x_i\|_1+1+ O\big(d^{1.5}\cdot \log d\cdot \log^{1.5}(1/\delta)/(\varepsilon n)\big)$ messages with each containing $O(\log d)$ bits. 
\end{enumerate}
\label{th:sparse_vector}
\end{restatable}

\begin{proof}
For privacy, invoking Lemma~\ref{lm:baseline} ensures that for each $j \in \{0,1,2,\dots,\log(d)\}$, $R_{\mathrm{cnt}}^{[2^j-1+2^j]}$ preserves $(\varepsilon/2,\delta/2)$-DP.
Additionally, for each $j\in \{0,1,2,\dots,\log(d)\}$, by combining Lemma~\ref{lm:baseline} with advanced composition and the fact that each $x_i\in D[2^{j-1}+1,2^j]$ affects at most $2^j$ number of $R_{\mathrm{sum}}^{[2^j-1+2^j],k}$, we have that, $\big\{R_{\mathrm{sum}}^{[2^j-1+2^j],k}\big\}_{k\in[d]}$ preserves $(\varepsilon/2,\delta/2)$-DP.  
Given that each $x_i$ impacts exactly one $R_{\mathrm{cnt}}^{[2^j-1+2^j]}$ and one $\big\{R_{\mathrm{sum}}^{[2^j-1+2^j],k}\big\}_{k\in[d]}$, the overall privacy guarantee is achieved.

Concerning utility, Theorem~\ref{lm:baseline} implies that for each $j\in\{0,1,2,\dots,\log d\}$, with probability at least $1-\frac{\beta}{2(\log d +1 )}$,  we have
\begin{align}
\nonumber
& \Big|\mathrm{Count}\big([2^{j-1}+1,2^j]\big)-\widetilde{\mathrm{Count}}\big([2^{j-1}+1,2^j]\big)\Big|
\\
\label{eq:th:sparse_vector_1}
\leq & \frac{2}{\varepsilon} \cdot \log\big(2(\log d+1)/\beta\big),
\end{align}
and the difference in sums, also with probability at least $1-\frac{\beta}{2(\log d +1 )}$, is well bounded:
\begin{align}
\nonumber
& \Big|\mathrm{Sum}\big([2^{j-1}+1,2^j]\big)-\widetilde{\mathrm{Sum}}\big([2^{j-1}+1,2^j]\big)\Big|_{\infty}
\\
\label{eq:th:sparse_vector_2}
= & O\Big(\sqrt{2^j\log(1/\delta)}\cdot \log(d/\beta)/\varepsilon\Big).
\end{align}
Aggregating these probabilities, we ensure both (\ref{eq:th:sparse_vector_1}) and (\ref{eq:th:sparse_vector_2}) hold for all $j$ with probability at least $1-\beta$.

(\ref{eq:th:sparse_vector_1}) implies that
\begin{equation}
\label{eq:th:sparse_vector_3}
j\leq \Big\lceil\log\big(\mathrm{Max}_{\ell_2}(D)\big)\Big\rceil.
\end{equation}
and 
\begin{align}
\nonumber
&\sum_{j=\log\tau+1}^{\log(d)} \mathrm{Count}\big(D[2^{j-1}+1,2^j]\big) 
\\
\nonumber
=& O\Big(\log d \log(\log d/\beta)/\varepsilon\Big),
\end{align}
which sequentially deduces
\begin{align}
\nonumber
 &\Big|\sum_{j=\log\tau+1}^{\log(d)}\mathrm{Sum}\big([2^{j-1}+1,2^j]\big)\Big|_{\infty}
 \\
\label{eq:th:sparse_vector_4}
 =&O\Big(\log d \log(\log d/\beta)/\varepsilon\Big).
\end{align}

Combining (\ref{eq:th:sparse_vector_2}) and (\ref{eq:th:sparse_vector_3}), we have
\begin{align}
\nonumber
& \Big|\sum_{j=0}^{\log\tau}\Big(\mathrm{Sum}\big([2^{j-1}+1,2^j]\big)-\widetilde{\mathrm{Sum}}\big([2^{j-1}+1,2^j]\big)\Big)\Big|_{\infty}
\\
\label{eq:th:sparse_vector_5}
= & O\Big(\mathrm{Max}_{\ell_2}(D)\cdot\sqrt{\log(1/\delta)}\cdot \log(d/\beta)/\varepsilon\Big).
\end{align}

Finally, combining (\ref{eq:th:sparse_vector_4}) and (\ref{eq:th:sparse_vector_5}) leads to our statement for utility.

For communication, recall that each $\mathrm{Baseline}$ without random rounding yields $1+O(\log(1/\delta)/(\varepsilon n))$ messages per user in expectation.  Combing this with facts that each $x_i$ has $|x_i|_1$ number of non-zero coordinates, and there is only one $[2^{j-1}+1,2^j]$ such that $|x_i|\in [2^{j-1}+1,2^j]$, we derive the desired statement.
One special note is that each message requires $O(\log d)$ bits to specify the dimension.
\end{proof}

\section{Practical Optimizations}
\label{sec:optimization}

In this section, we briefly discuss some practical optimizations for our protocols, although they do not affect the asymptotic results.

As mentioned, for sum estimation protocol of \cite{balle2020private} attains an error very closely to that of \cite{ghazi2021differentially}.
Meanwhile, \cite{balle2020private} send $O(1)$ messages per user while  \cite{ghazi2021differentially} sends $1+o(1)$ messages.  Although the former is asymptotically smaller, the $o(1)$ term, or $O\big(\log^2(n)\cdot\log(1/\delta)/(\varepsilon\sqrt{n})\big)$ to be more precise, is actually not negligible for $n$ not too large. 
Since our mechanism uses sum estimation as a black box, in our implementation we choose either \cite{ghazi2021differentially} or \cite{balle2020private} based on the concrete values of $n,\varepsilon,\delta$.

Furthermore, recall that in $\mathrm{SumDP}$, we invoke $\log(U)+1$ instances of $\mathrm{BaselineSumDP}$, corresponding to different domain sizes $1,2,4,\dots,U$.
We note that using the protocol of \cite{ghazi2021differentially} without random rounding yields a message number of $1+O(U\log^2(U)\log(U/\delta)/(n\varepsilon))$, which may be better than doing a random rounding when the domain size is small.
Therefore, for different domain sizes, we adopt different baselines: \cite{ghazi2021differentially} with or without random rounding or \cite{balle2020private}. We again choose the best one based on the concrete values of  $n$, $\varepsilon$, $\delta$, and domain size. 

\section{Experiments}

In addition to the improved asymptotic results, we have also conducted experiments comparing our protocols with the previous algorithms.

\textbf{Sum estimation:}
Our $\mathrm{SumDP}$ mechanism was evaluated alongside two baselines: $\mathrm{GKMPS}$ \cite{ghazi2021differentially} and $\mathrm{BBGN}$ \cite{balle2020private}. 
We also compared its error to the state-of-the-art central-DP mechanism~\cite{dong2023universal} as a gold standard.
The two-round protocol from \cite{huang2021instance} is solely a theoretical result. It not only has a large message number and errors but also has an impractical running time. More precisely, \cite{huang2021instance} applies the method from \cite{ghazi2021power} to approximate $\mathrm{Max}(D)$ to obtain a clipping threshold $\tau$.    
In \cite{ghazi2021power}, each randomizer has a computation of $O(n\log^2 U)$. 
Additionally, both \cite{huang2021instance} and \cite{ghazi2021power} only present their theoretical results without any concrete implementation.

\textbf{High dimensional sum:}
For high-dimensional sum estimation, our $\mathrm{HighDimSumDP}$ was compared against the one-round protocol $\mathrm{HLY}$ proposed in \cite{huang2021instance}. Similar to sum estimation, the two-round protocol from \cite{huang2021instance} faced efficiency challenges, as outlined before.  For this problem, we use the central-DP mechanism in \cite{dong2023better} as the gold standard.

\textbf{Sparse vector aggregation:}
For sparse vector aggregation, we assessed $\mathrm{SparVecSumDP}$ against $\mathrm{NaiveVecSumDP}$, which uses the dimension $d$ as the upper bound for sparsity.  Here, we also use the central-DP mechanism in \cite{dong2023better} as the gold standard.

\subsection{Setup}

\textit{Datasets.} We used both synthetic and real-world datasets in the experiments.  For sum estimation, the synthetic data was generated from two families of distributions over $[U]$ with $n = U= 10^5$: 
Zipf distribution $f(x)\propto  (x+a)^{-b}$ with $a = 1$, $b = 3$ and $a = 1$, $b = 5$; Gauss distribution with $\mu = 5$, $\sigma = 5$ and $\mu = 50$, $\sigma = 50$.  The real-world datasets were collected from Kaggle, including San-Francisco-Salary (SF-Sal)~\cite{sf_salaries_2011_2014}, Ontario-Salary (Ont-Sal)~\cite{ontario_dataset_2020}, Brazil-Salary (BR-Sal)~\cite{brazil_public_worker_salary_2020}, and Japan-trade (JP-Trad)~\cite{million_data_csv_2020}.
SF-Sal, BR-Sal, and Ont-Sal are salary data from San Francisco, Brazil, and Ontario for the years 2014, 2020, and 2020, respectively, with amounts presented in thousands of US dollars (K USD).
For the salary data, we set the domain limit $U$ to $2.5\times 10^5$, which is the world's highest recorded salary.  
The JP-Trad dataset, capturing Japan's trade statistics from 1988 to 2019, includes 100 million entries. 
We selected a subset of approximately 200,000 tuples, covering Japan's trade activities with a designated country.  We set $U$ as the maximum value across the entire dataset.  This dataset also has the amounts expressed in K USD.
The details of these real-world data can be found in Table~\ref{tab:real-world-datasets}.

\begin{table}[h]
    \centering
    \resizebox{0.5\linewidth}{!} 
    {
	\renewcommand\arraystretch{1.25}
    \begin{tabular}{c|c|c|c}
        \hline
        Dataset & $n$ & $U$& $\mathrm{Max}(D)$
        \\
        \hline
        \hline
        SF-Sal & $1.49\times 10^5$ & $2.5\times 10^5$ & $568$
        \\
        \hline
       Ont-Sal &$5.75\times 10^5$ & $2.5\times 10^5$ & $1750$
       \\
        \hline
         BR-Sal & $1.09\times 10^6$ & $2.5\times 10^5$ & $343$
        \\
        \hline
        JP-Trad & $1.81\times 10^5$ & $2\times 10^5$ & $2810$
        \\
    \hline
    \end{tabular}
    }
    \caption{Real-world datasets used in sum aggregation}
    \label{tab:real-world-datasets}
\end{table}

For high-dimensional sum estimation, we utilized the MNIST dataset~\cite{mnist_kaggle_2020}, comprising 70,000 digit images, with each
represented by a vector of dimension $d = 28\times 28 = 784$.  The $U_{\ell_2}$ parameter was set to $2^{10}$ according to \cite{huang2021instance}.

The sparse vector aggregation experiments were conducted using the AOL-user-ct-collection (AOL)~\cite{aol}, documenting 500,000 users' clicks on 1,600,000 URLs.
we consolidated every 100 webpages into a single dimension, resulting in a dimensionality of $1.6 \times 10^4$, and selected the first 50,000 users as our testing dataset.

\textit{Experimental parameters.} All experiments are conducted on a Linux server equipped with a 24-core 48-thread 2.2GHz Intel Xeon CPU and 256GB memory.  We used absolute error, $\ell_2$ error, and $\ell_{\infty}$ error metrics for sum estimation, high-dimensional sum, and sparse vector aggregation respectively. 
We repeated each experiment 50 times, discarding the 10 largest and smallest errors for an averaged result from the remaining 30.
The message complexity was quantified by the average number of messages per user, with each message containing $O(\log(d) + \log(U) + \log(n))$ bits. 
For the privacy budget, we used $\varepsilon = 0.2, 1, 5$, and the default value was set to $1$. 
To protect data privacy, $\delta$ should be set to a value significantly smaller than the inverse of the data size. Therefore, $\delta$ was fixed at $10^{-12}$ in our experiments.
\footnote{Notably, for our mechanism, a larger $\delta$ will not affect error but will benefit the communication, albeit minimally as it affects only the logarithmic term.}
The failure probability $\beta$ was set at 0.1.

\begin{table*}
    \centering
    \resizebox{1\linewidth}{!} 
    {
	\renewcommand\arraystretch{1.15}
	\begin{tabular}{c|c|c||c|c|c|c||c|c|c|c}
        \hline
        \multicolumn{3}{c||}{\multirow{4}{*}{Dataset}} &  \multicolumn{4}{c||}{Simulated Data} & \multicolumn{4}{c}{Real-world Data} 
        \\
        \cline{4-11}
        \multicolumn{3}{c||}{} & \multicolumn{2}{c|}{Zipf} & \multicolumn{2}{c||}{Gauss} & \multirow{3}{*}{SF-Sa} & \multirow{3}{*}{Ont-Sa}  &\multirow{3}{*}{BR-Sa}  & \multirow{3}{*}{JP-Trad} 
        \\
        \cline{4-7}
        \multicolumn{3}{c||}{} & $a=1$  & $a=1$ & $\mu=5$ & $\mu=50$ & & & &
        \\
        \multicolumn{3}{c||}{} & $b=3$ & $b=5$ & $\sigma=5$ & $\sigma=50$ & & & &
        \\
        \hline
        \multirow{8}{*}{\makecell[c]{$1$-D \\ Sum}}& 
        \multicolumn{2}{c||}{\multirow{2}{*}{\makecell[c]{SOTA under \\central-DP RE(\%)}}} & \multirow{2}{*}{{0.53}}&\multirow{2}{*}{{0.0247}} & \multirow{2}{*}{{0.00921}}& \multirow{2}{*}{{0.00737}}& \multirow{2}{*}{{0.00936}}& \multirow{2}{*}{{0.0028}}& \multirow{2}{*}{{0.0452}}&\multirow{2}{*}{{0.168}}
        \\
        &\multicolumn{2}{c||}{} & & & & & & & &
        \\
        \cline{2-11}
        & \multirow{2}{*}{SumDP   (Ours)} & RE(\%) & \textbf{{1.13}} & \textbf{{0.0661}} &  \textbf{{0.00351}} & \textbf{{0.00452}} & \textbf{{0.00989}} & \textbf{{0.0075}} & \textbf{{0.0249}} & \textbf{{0.101}}
        \\
        & \multirow{2}{*}{} & \#Messages/user& {140} & {140} & {139} & {139} & {143} & {126} & {119} & {140}
        
        \\
        \cline{2-11}
        & \multirow{2}{*}{GKMPS} & RE(\%) & {54.5} & {96.3} & {22.5} & {3.11} & {2.44} & {0.372} & {9.26} & {9.02}
        \\
        & \multirow{2}{*}{} & \#Messages/user &{14200} & {14300} & {14300} & {14300} & {12200}& {7770} & {5950} & {11200}
        
        \\
        \cline{2-11}
        & \multirow{2}{*}{BBGN} & RE(\%) &{53} &{76.6} & {17.6} & {1.43} &  {1.96} & {0.294} & {3.53} & {4.46}
        \\
        & \multirow{2}{*}{} & \#Messages/user &\textbf{{9}} &\textbf{{9}} & \textbf{{9}}&\textbf{{9}} & \textbf{{9}} & \textbf{{8}} & \textbf{{8}} & \textbf{{9}}
        \\
        \hline
        \end{tabular}
    }
    \caption{Comparison among sum estimation mechanisms under shuffle-DP ($\varepsilon=1$). RE denotes the relative error.}
    \label{tab:results}
\end{table*}

\subsection{Experimental Results for Sum Aggregation}

\begin{figure*}[t]
\centering
\includegraphics[width=1\textwidth]{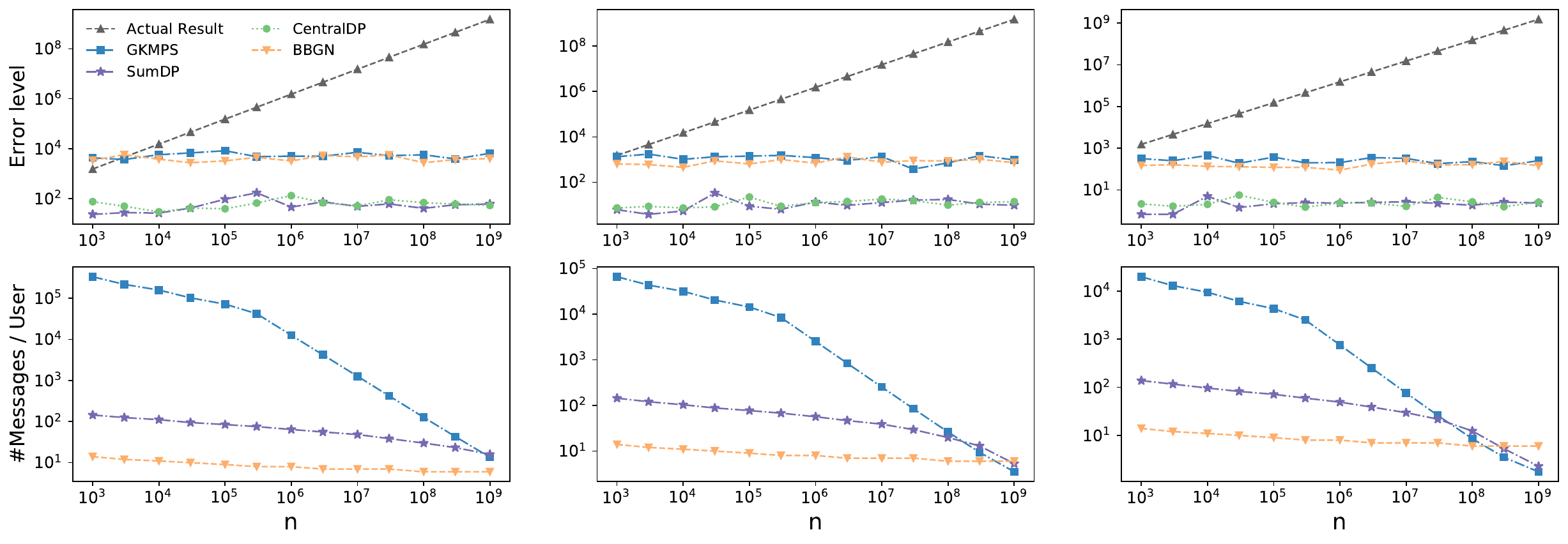}
    \caption{Error levels and average messages per user for the sum estimation mechanisms under shuffle-DP with different data size $n$. $\mathrm{CentralDP}$ represents the state-of-the-art algorithm for sum estimation under central-DP.}
    \label{fig:datasize}
\end{figure*}

\textit{Utility and communication.} 
Table~\ref{tab:results} shows the errors and the average number of messages per user across various mechanisms for sum estimation under shuffle-DP over both simulated and real-world data.
The results indicate a clear superiority of $\mathrm{SumDP}$ in terms of utility.  
$\mathrm{SumDP}$ consistently maintains an error below $2\%$ across all eight tests, further reducing it to below $0.2\%$ in seven cases.
In contrast, $\mathrm{GKMPS}$ and $\mathrm{BBGN}$ exhibit significantly higher error levels.
Our improvement over $\mathrm{GKMPS}$ and $\mathrm{BBGN}$ can be up to $3000\times$.  
This superiority is particularly evident in the JP-Trad dataset, where $\mathrm{SumDP}$ surpasses $\mathrm{GKMPS}$ and $\mathrm{BBGN}$ by more than $40\times$ even with a pre-established $U$ based on strong prior knowledge.
This validates our theoretical analysis:
$\mathrm{SumDP}$ achieves an instance-specific error, unlike $\mathrm{GKMPS}$ and $\mathrm{BBGN}$, which target worst-case errors. 
Furthermore, we observe that $\mathrm{SumDP}$ attains error levels similar to the gold standard, and produces even smaller errors in about half of the cases.
This is because while the two methods have the same asymptotic error bounds, they are both upper bounds that may not be tight (in constant factors) on all instances. Therefore, the actual error of either mechanism could be smaller than the other.

In terms of communication, neither $\mathrm{SumDP}$ nor $\mathrm{GKMPS}$ achieves the theoretical ideal of single-message communication per user in all tests.  In contrast, $\mathrm{BBGN}$ requires fewer messages.
This is because, even though $\mathrm{SumDP}$ and $\mathrm{GKMPS}$ theoretically reach $1+o(1)$ messages per user, the term $o(1)$ masks substantial logarithmic factors, leading to significantly higher actual message counts, especially when $n$ is small. 
In contrast, $\mathrm{BBGN}$ maintains constant messages per user.
Later, we will show that as $n$ increases, both $\mathrm{GKMPS}$ and $\mathrm{SumDP}$ exhibit a trend towards achieving a single-message communication per user.
Additionally, $\mathrm{SumDP}$ requires much fewer messages than $\mathrm{GKMPS}$.  This is attributed to the optimization described in Section~\ref{sec:optimization}, where our mechanism intelligently chooses the more communication-efficient method between $\mathrm{GKMPS}$ and $\mathrm{BBGN}$.

\begin{figure*}[t]
\centering
\includegraphics[width=1\textwidth]{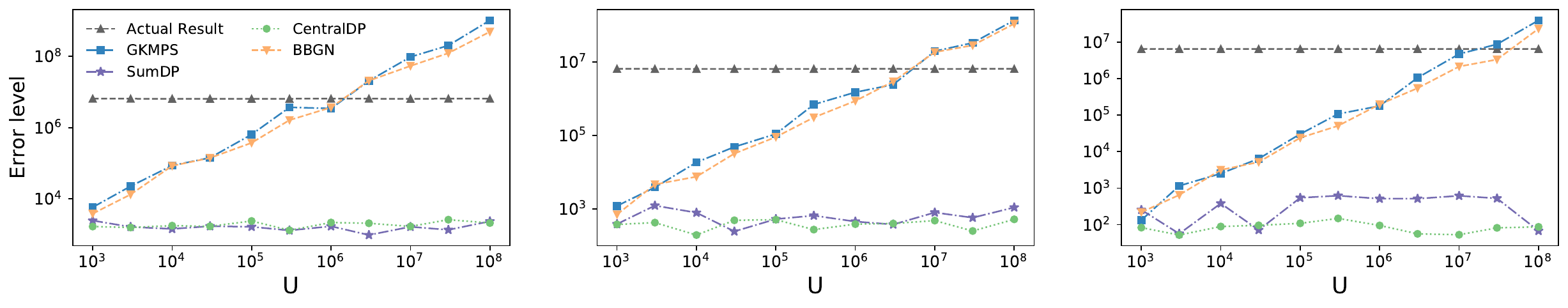}
    \caption{Error levels of the mechanisms for sum estimation under shuffle-DP with different value domain $U$.}
    \label{fig:domain}
\end{figure*}

\begin{figure*}[t]
\centering
\includegraphics[width=1\textwidth]{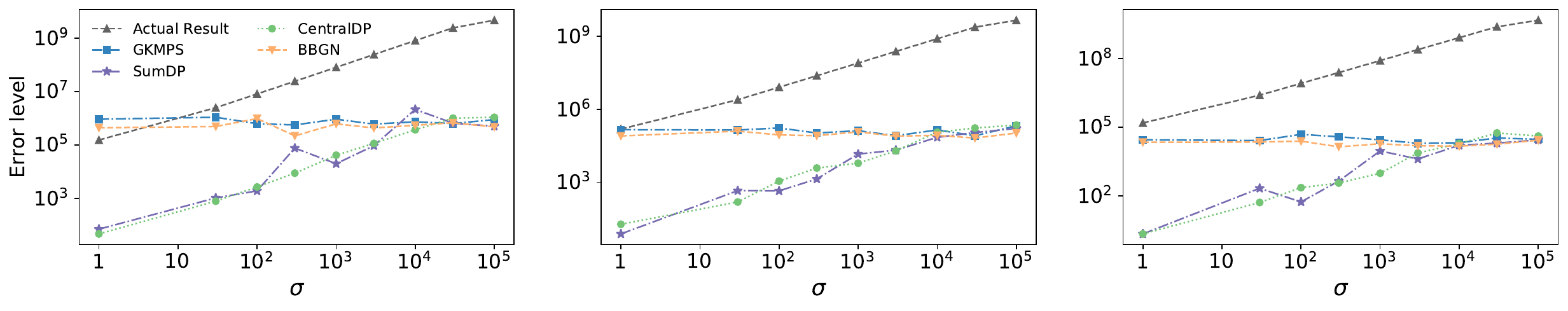}
    \caption{Error levels of the mechanisms for sum estimation under shuffle-DP with data drawn from a Gaussian distribution with different $\sigma$.}
    \label{fig:skewness}
\end{figure*}

\textit{Data size.} 
To assess the impact of varying data sizes, we conducted experiments using simulated data generated from a Gaussian distribution with $\mu = 1$, $\sigma = 1$, and domain size $U = 10^3$.
The data size varied from $10^3$ to $10^9$, and we tested with different privacy budgets $\varepsilon=0.2,1,5$.
The error levels and average messages per user are depicted in Figure~\ref{fig:datasize}.  Note that in all our figures, both axes are in log-scale and the actual query results are plotted alongside the error levels to provide a benchmark for assessing the utility of the mechanisms.
In terms of utility, $\mathrm{SumDP}$ consistently has a high utility even with small $n$ and $\varepsilon$, akin to the state-of-the-art central-DP mechanism.
Notably, the error levels for all mechanisms did not exhibit significant changes with varying $n$, matching our analytical analyses that the errors in $\mathrm{BBGN}$ and $\mathrm{GKMPS}$ are dependent on $U$, while the errors in $\mathrm{SumDP}$ and the central-DP mechanism depend on $\mathrm{Max}(D)$, all of which are not directly influenced by the data size.

Regarding communication, $\mathrm{GKMPS}$ showed a decrease in the average messages per user with larger $n$ values, while $\mathrm{BBGN}$ maintained a constant message complexity.
$\mathrm{SumDP}$ displayed a unique trend: it maintained its message complexity for smaller $n$ values, then gradually decreased it as $n$ increased. 
This pattern is attributed to $\mathrm{SumDP}$ initially leveraging $\mathrm{BBGN}$ for smaller datasets and then transitioning to $\mathrm{GKMPS}$ for larger datasets.
Moreover, as $n$ increases, both $\mathrm{GKMPS}$ and $\mathrm{SumDP}$ demonstrate a progression towards single-message communication per user, aligning with our theoretical analysis that both mechanisms achieve $1+o(1)$ messages per user.

\textit{Domain size and data skewness.}
To investigate the influence of domain size $U$ and data skewness on the error level of various mechanisms, we utilized simulated data with $n=10^5$, drawn from a Gaussian distribution with $\mu = 50$, $\sigma = 50$, and a variable domain size $U$ ranging from $10^3$ to $10^8$ and a Gaussian distribution with $\mu = 1$, $U=10^5$, and a varying $\sigma$ from $1$ to $10^5$. The results are plotted in Figure~\ref{fig:domain} and~\ref{fig:skewness}.
The first message the figures convey is nothing but a reconfirmation that $\mathrm{GKMPS}$ and $\mathrm{BBGN}$ have error proportional to $U$ and $\mathrm{DPSum}$ have the error proportional to the actual maximum value in the dataset. Moreover, Figure~\ref{fig:skewness} shows in the cases when the actual maximum equals to $U$, $\mathrm{DPSum}$ achieves an error level nearly the same as the $\mathrm{GKMPS}$ and $\mathrm{BBGN}$.

\begin{table*}
\centering
\resizebox{0.75\linewidth}{!} 
{
	\renewcommand\arraystretch{1.15}
	\begin{tabular}{c|c|c|c|c}
        \hline
       STOTA under & \multicolumn{2}{c|}{$\mathrm{HighDimSumDP}$} &  \multicolumn{2}{c}{$\mathrm{HLY}$}
        \\
        \cline{2-5}
        central-DP RE(\%) & RE(\%) & \#Messages/user & RE(\%) & \#Messages/user
        \\
        \hline
        {0.328} & {6.43} & {446,000} & {1070} & {986,000,000}
        \\
        \hline
        \end{tabular}
}
\caption{Comparison among mechanisms for high-dimensional sum estimation under shuffle-DP ($\varepsilon=5$). RE denotes the relative error, which is $\ell_2$ Error/$\|\mathrm{Sum}(D)\|_2$.}
\label{tab:results_high}
\end{table*}

\begin{table*}
\centering
\resizebox{0.75\linewidth}{!} 
{
	\renewcommand\arraystretch{1.15}
	\begin{tabular}{c|c|c|c|c}
        \hline
       STOTA under & \multicolumn{2}{c|}{$\mathrm{SparVecSumDP}$} &  \multicolumn{2}{c}{$\mathrm{NaiveVecSumDP}$}
        \\
        \cline{2-5}
        central-DP RE(\%) & RE(\%) & \#Messages/user & RE(\%) & \#Messages/user
        \\
        \hline
        {3.2} & {4.4} & {4,970,000} & {21} & {1,450,000}
        \\
        \hline
        \end{tabular}
}
\caption{Comparison among protocols for sparse vector aggregation under shuffle-DP ($\varepsilon=5$).  RE denotes the relative error, which is $\ell_{\infty}$ Error/$|\mathrm{Sum}(D)|_{\infty}$. }
\label{tab:results_sparse}
\end{table*}

\subsection{Experimental Results for High-dimensional Sum and Sparse Vector Aggregation}

The results for high-dimensional sum and sparse vector aggregation are shown in table~\ref{tab:results_high} and~\ref{tab:results_sparse} respectively. 
In these experiments, we opted for a higher privacy budget, setting $\varepsilon$ to 5.
That is because answering high-dimensional queries is much more complex and we need larger $\varepsilon$ to guarantee the utility~\cite{huang2021instance}.  
First, similar to the results of sum estimation, our mechanisms for high-dimensional and sparse vector aggregation yield an error level, that is close to the state-of-the-art mechanisms under central-DP and is much smaller than the prior works with worst-case optimal error.  Second, our mechanisms achieve a similar or even better communication than prior works.

\section{Conclusion}
In this paper, we study answering sum estimation under the shuffle-DP model, where prior works either only achieve worst-case optimal error or have very a heavy communication cost. We introduce the first protocol that not only has instance-optimal error but also achieves optimal communication efficiency, i.e., requiring only $1+o(1)$ messages per user. 
Furthermore, we successfully extend our technique to address high-dimensional sum estimation and sparse vector aggregation.
Finally, we would like to mention two interesting directions for future research. The first is how to extend our domain division technique to private sum estimation in various models, such as the multi-party secure computation model.
Besides, since the private summation is the foundation to private protocols for various machine learning models, investigating its potential to enhance utility in these advanced tasks would also be valuable.

\section*{Acknowledgements}
This work has been in part supported by a grant from ONR, a grant from the DARPA SIEVE program under a subcontract from SRI, a Packard Fellowship, and contributions from Intel, Bosch, and Cisco. Additionally, this work has been funded by NSF awards under grant numbers 2128519, 2044679, 2338772, and 2148359.
Qiyao Luo and Ke Yi have been supported by HKRGC under grants 16205420, 16205422, and 16204223.

\bibliographystyle{alpha}
\bibliography{main}
\end{document}